\documentclass[a4paper,onecolumn,accepted=2021-12-02]{quantumarticle}
\pdfoutput=1
\usepackage{mathrsfs} %% use the command \mathscr{…}
\usepackage{bbm} % to use the command \mathbbm{1}
\usepackage{bm}  % \bm{x}

\usepackage{subcaption}
\usepackage{caption}
\usepackage[dvipsnames]{xcolor}
\definecolor{beamer@blendedblue}{rgb}{0.2,0.2,0.7}      % color used in beamer
\usepackage{framed}
\definecolor{shadecolor}{rgb}{0.9,0.9,0.9}
\usepackage{mathtools}
\usepackage{amsmath}
\allowdisplaybreaks
\usepackage[shortlabels]{enumitem}
\setlist{itemsep=0pt,topsep=2pt,partopsep=1pt}
\usepackage[titletoc,title]{appendix}
\usepackage{graphicx,epic,eepic,epsfig,amsmath,latexsym,amssymb,verbatim,color}
\usepackage{dsfont}

\usepackage{booktabs}
\usepackage{multirow}

\usepackage{authblk} % package for author list
\usepackage{tcolorbox}
\usepackage{relsize}

\usepackage{float}
\usepackage{tikz}
\usepackage{hyperref}

\usepackage{bookmark}

\usepackage{url}
\usepackage{theorem}
\newtheorem{definition}{Definition}

\newtheorem{lemma}[definition]{Lemma}

\newtheorem{theorem}[definition]{Theorem}
\newtheorem{corollary}[definition]{Corollary}

\newtheorem{property}[definition]{Property}

\newcommand\ce[1]{\nu(#1)}

\newcommand{\vb}{\vspace*{\baselineskip}} %% sometimes we just need to introduce some extra vertical free space
\newcommand\JJQ[1]{{\color{black} #1}}

\usepackage{pifont}% http://ctan.org/pkg/pifont
\definecolor{Gray}{gray}{0.92}
\definecolor{Gray2}{gray}{0.75}
\definecolor{maroon}{cmyk}{0,0.87,0.68,0.32}

\def\squareforqed{\hbox{\rlap{$\sqcap$}$\sqcup$}}
\def\qed{\ifmmode\squareforqed\else{\unskip\nobreak\hfil
\penalty50\hskip1em\null\nobreak\hfil\squareforqed
\parfillskip=0pt\finalhyphendemerits=0\endgraf}\fi}
\def\endenv{\ifmmode\;\else{\unskip\nobreak\hfil
\penalty50\hskip1em\null\nobreak\hfil\;
\parfillskip=0pt\finalhyphendemerits=0\endgraf}\fi}
\newenvironment{proof}{\noindent \textbf{{Proof~}}}{\hfill $\blacksquare$}

\newcounter{remark}
\newenvironment{remark}[1][]{\refstepcounter{remark}\par\medskip\noindent%
\textbf{Remark~\theremark #1} }{\medskip}

\newcounter{example}

\mathchardef\ordinarycolon\mathcode`\:
\mathcode`\:=\string"8000
\def\vcentcolon{\mathrel{\mathop\ordinarycolon}}
\begingroup \catcode`\:=\active
  \lowercase{\endgroup
  \let :\vcentcolon
  }

\DeclareFontFamily{U}{mathx}{\hyphenchar\font45}
\DeclareFontShape{U}{mathx}{m}{n}{<-> mathx10}{}
\DeclareSymbolFont{mathx}{U}{mathx}{m}{n}
\DeclareMathAccent{\widebar}{0}{mathx}{"73}

\newcommand{\wt}[1]{\widetilde{#1}}

\ifcsname{ul}\endcsname%
\else%
\fi%

\renewcommand{\iff}{\textit{iff.~}}
\newcommand{\sbar}{\;\rule{0pt}{9.5pt}\middle|\;}

\definecolor{darkblue}{RGB}{0,76,156}
\definecolor{darkkblue}{RGB}{0,0,153}
\definecolor{blue2}{RGB}{102,178,255}
\definecolor{darkred}{RGB}{195,0,0}

\let\emptyset\varnothing
\let\leq\leqslant
\let\geq\geqslant
\DeclareFontFamily{U}{mathb}{\hyphenchar\font45}
\DeclareFontShape{U}{mathb}{m}{n}{
<-6> mathb5 <6-7> mathb6 <7-8> mathb7
<8-9> mathb8 <9-10> mathb9
<10-12> mathb10 <12-> mathb12
}{}
\DeclareSymbolFont{mathb}{U}{mathb}{m}{n}
\DeclareMathSymbol{\succneq}{\mathrel}{mathb}{"CF}

\newcommand{\ket}[1]{\left\vert #1 \right\rangle}
\newcommand{\bra}[1]{\left\langle #1 \right\vert}
\newcommand{\ketbra}[2]{\vert#1\rangle\!\langle#2\vert}
\newcommand\proj[1]{\vert #1 \rangle\!\langle #1 \vert}
\newcommand{\linear}[1]{\mathscr{L}(#1)}
\newcommand{\hptp}[1]{\mathscr{T}^\dagger(#1)}
\newcommand{\herm}[1]{\mathscr{L}^\dagger(#1)}
\newcommand{\pos}[1]{\mathscr{P}(#1)}

\newcommand{\channel}[1]{\mathscr{C}(#1)}
\newcommand{\subchannel}[1]{\mathscr{C}_{\bullet}(#1)}

\newcommand{\ox}{\otimes}
\newcommand{\1}{I}
\DeclareMathOperator{\tr}{Tr}  % trace
\newcommand{\id}{\operatorname{id}}

\newcommand{\opn}[1]{\operatorname{#1}}
\newcommand{\norm}[2]{\ensuremath{\left\lVert#1\right\rVert_{#2}}}%
\newcommand*{\cA}{\mathcal{A}}

\newcommand*{\cD}{\mathcal{D}}

\newcommand*{\cF}{\mathcal{F}}

\newcommand*{\cH}{\mathcal{H}}
\newcommand*{\cI}{\mathcal{I}}

\newcommand*{\cK}{\mathcal{K}}

\newcommand*{\cM}{\mathcal{M}}
\newcommand*{\cN}{\mathcal{N}}
\newcommand*{\cO}{\mathcal{O}}
\newcommand*{\cP}{\mathcal{P}}
\newcommand*{\cQ}{\mathcal{Q}}
\newcommand*{\cR}{\mathcal{R}}
\newcommand*{\cS}{\mathcal{S}}
\newcommand*{\cT}{\mathcal{T}}
\newcommand*{\cU}{\mathcal{U}}
\newcommand*{\cV}{\mathcal{V}}

\newcommand*{\cX}{\mathcal{X}}

\newcommand*{\bC}{\mathbb{C}}

\newcommand*{\bR}{\mathbb{R}}

\newcommand*{\bZ}{\mathbb{Z}}

\usetikzlibrary{shadows.blur}
\usetikzlibrary{shapes.symbols}
\usetikzlibrary{intersections,
                positioning,
                shapes}
\pgfdeclarelayer{shadow}
\pgfsetlayers{shadow,main}
\usetikzlibrary{shadows.blur}
\usetikzlibrary{shapes.symbols}
\usetikzlibrary{intersections,
                positioning,
                shapes}
\usetikzlibrary{decorations.pathreplacing}

\tikzset{%
   Channel/.style =
    { thick, %
      draw = black, %
      fill = white, %
      align=center, %
      rectangle, %
      rounded corners, %
    }, %
    CWire/.style =
    { line width=0.5pt, %
      line cap=round,
      double, %
      double distance = 1.2pt, %
      -, %
    },
    QWire/.style =
    { line width=2.5pt, %
      line cap=round,
      -, %
    },
}

\usepackage{thmtools}

\usepackage[numbers,sort&compress]{natbib}
\begin{document}
\title{{Physical Implementability of Linear Maps 
        and Its Application in Error Mitigation}}

\author{Jiaqing Jiang}
\affiliation{Institute for Quantum Computing, Baidu Research, Beijing 100193, China}
\affiliation{Computing and Mathematical Sciences, California Institute of Technology, Pasadena, CA USA}

\author{Kun Wang}
\affiliation{Institute for Quantum Computing, Baidu Research, Beijing 100193, China}

\author{Xin Wang}
\email{wangxin73@baidu.com}
\affiliation{Institute for Quantum Computing, Baidu Research, Beijing 100193, China}

\maketitle

\begin{abstract}
Completely positive and trace-preserving maps characterize physically implementable quantum operations.  On the other hand, general linear maps, such as positive but not completely positive maps, which can not be physically implemented, are fundamental ingredients in quantum information, both in theoretical and practical perspectives.
This raises the question of how well one can simulate or approximate the action of a general linear map by physically implementable operations. 
In this work, we introduce a systematic framework to resolve this task using the quasiprobability decomposition technique. We decompose a target linear map into a linear combination of physically implementable operations
and introduce the \emph{physical implementability} measure as the least amount of negative portion that the quasiprobability must pertain, which directly quantifies the cost of simulating a given map using physically implementable quantum operations. We show this measure is efficiently computable by semidefinite programs and prove several properties of this measure, such as faithfulness,
additivity, and unitary invariance. We derive lower and upper bounds in terms of the Choi operator's trace norm and obtain analytic expressions for several linear maps of practical interests. Furthermore, we endow this measure with an operational meaning within the quantum error mitigation scenario: 
it establishes the lower bound of the sampling cost achievable via the quasiprobability decomposition technique.
In particular, for parallel quantum noises, we show that global error mitigation has {no advantage} over local error mitigation.
\end{abstract}

%{
%  \hypersetup{linkcolor=black}
%  \pdfbookmark[section]{\contentsname}{toc}
%  \tableofcontents
%}

\section{Introduction}

% The first paragraph introduces the importance of quantum channels/evolutions.

The postulates of quantum mechanics prescribe that
the evolution of a closed global quantum system must be unitary~\cite{nielsen2011quantum}.
The physically implementable quantum operations are then obtained in the reduced dynamics of subsystems
and are mathematically characterized by completely positive and trace-preserving
maps (CPTPs)~\cite{kraus1983states}.
Nevertheless, many other linear maps such as positive but not completely positive maps,
which are impossible to be physically implemented, are also fundamental ingredients
from theoretical and practical perspectives.
On the one hand, positive maps play an essential role in quantum information processing;
for any entangled state, there exists a positive map that determines whether or not the
given state is entangled~\cite{horodecki1996Separability}.
On the other hand, under certain conditions, the reduced dynamics
might not be captured within the completely positive map formalism~\cite{pechukas1994reduced,carteret2008dynamics}.
Thus we have to relax the completely positivity condition to less conservative ones such as positive maps.
\JJQ{In summary, these exceptions witness the importance of positive maps in quantum theory
and motive us to raise the following fundamental problem: 
\begin{center}
\textit{How to simulate the action of `non-physical' linear maps using physical operations?}
\end{center}}

\JJQ{The Structural Physical Approximation (SPA)
method~\cite{horodecki2003limits,horodecki2002method,fiuravsek2002structural,korbicz2008structural} offers a
structural way to resolve one important case, i.e., approximating the
positive but not completely positive maps (non-physical) using completely positive maps (physical).}
 Briefly, SPA performs a convex mixture of the original positive map with the completely depolarizing channel, where the completely depolarizing channel  projects the input state onto the maximally
mixed state, i.e., $\Omega(\rho) = \1/d, \forall\rho$. Such a mixture is always feasible
since \emph{any} linear map (not necessarily positive) mixed with a sufficiently
large amount of $\Omega$ will result in a completely positive map~\cite{zyczkowski1998volume,gurvits2002largest}.
\JJQ{However, to make the SPA efficient, one should mix as little amount of $\Omega$ as possible.}
%However, for efficiency concerns, one should use as little  $\Omega$ as possible.%
We call the completely positive map generated by the least $\Omega$
the SPA of the original map.
The SPA method finds novel applications in entanglement detection
%since it
%yields some physical process that is experimentally implementable~\cite{lim2011experimental}
by yielding experimentally implementable physical processes for positive maps that are non-physical.

% We introduce the measure - physically implementability
\JJQ{In this paper, we introduce a systematic framework to resolve the above fundamental problem and establish an operational and quantitative study of the physical implementability (or non-physicality from another perspective) of general linear maps, using the quasiprobability decomposition technique~\cite{crowder2015linearization,temme2017error,howard2017application,endo2018practical,takagi2020optimal} and theoretical tools from semidefinite programming.}
More specifically, we decompose the target linear map, 
i.e., Hermitian- and trace-preserving maps (not necessarily positive),
into a linear combination of physically implementable quantum operations, i.e., CPTPs.
Then the \emph{physical implementability} of the target map is defined to be
the least amount of negative portion that the quasiprobability must pertain.
To some extent, one can view our method as a generalization of the SPA method
by enlarging the resource set from the completely depolarizing channel to any CPTPs.
This physical implementability measure
is efficiently computable via semidefinite programming~\cite{Vandenberghe1996}. Besides, it
 possesses many desirable properties,
such as faithfulness, additivity with respect to tensor products, unitary channel invariance,
and monotonicity under superchannels.
For general target maps, we derive lower and upper bounds in terms of the Choi operator's trace norm. This bound is tight in the sense of equalities saturated by certain channels. For several linear maps of practical interests, we obtain analytic
expressions. \JJQ{Notably, by considering the inverse of a quantum channel as the target linear map, we endow this measure with an operational interpretation within the quantum error mitigation scenario as it quantifies the lower bound of the sampling cost achievable via the quasiprobability decomposition technique. This result establishes the  limits and delivered meaningful insights to quantum error mitigation schemes using the quasiprobability decomposition method.}

% Outline of the paragraph

\paragraph*{Outline and main contributions}

The outline and main contribution of this paper can be summarized as follows:
\begin{itemize}
  \item In Section~\ref{sec:preliminaries},
        we set the notation and formally define the concept of invertible linear maps.
  \item In Section~\ref{sec:physical implementability},
        we introduce the \emph{physical implementability} measure to characterize
        how well a linear map can be physically implemented or approximated.
        We show this measure is efficiently computable by semidefinite programs.
        We prove several properties of this measure, such as faithfulness, additivity,
        and unitary channel invariance.
        What's more, we derive bounds in terms of the Choi operator's trace norm
        and obtain analytic expressions for several linear maps of practical interests.
  \item In Section~\ref{sec:error mitigation},
        we enrich the proposed physical implementability measure \JJQ{with}
        an operational interpretation within the quantum error mitigation framework.
        It %quantifies the ultimate sampling cost achievable using the full expressibility of quantum computers.
        establishes the lower bound of the sampling cost achievable  via the quasiprobability decomposition technique.  For parallel quantum noises, we prove that global error mitigation has \textit{no advantage} over local error mitigation,
        i.e., dealing with quantum noises individually.
  \item In Appendix~\ref{sec:robustness}, we introduce the \emph{robustness} of physical implementability
        from the resource theory perspective.
        We discuss the relationship between two seemingly different quantities:
        physical implementability and robustness.
        We show that these two quantities
        are  equivalent in some sense.
\end{itemize}

%%%%%%%%%%%%%%%%%%%%%%%%%%%%%%%%%%%%%%%%%%%%%%%%%%%%%%%%%%%%%%%%%%%%%%%%%%%%
\section{Preliminaries}\label{sec:preliminaries}

In this section, we set the notations and define quantities that will be used throughout this paper.

\subsection{Notations}
We label different quantum systems by capital Latin letters (e.g., $A,B,R$).
The corresponding Hilbert spaces of these quantum systems 
are denoted as $\cH_A,\cH_B,\cH_R$, respectively. Throughout this paper, we only consider quantum systems of finite dimensions.
Systems with the same letter are assumed to be isomorphic: $A'\cong A$.
Multipartite quantum systems are described by tensor
product spaces, i.e., $\cH_{AB}=\cH_A\ox\cH_B$.
We use these labels as subscripts or superscripts to indicate which system the corresponding
mathematical object belongs to if necessary.
We drop the scripts when they are evident from the context.

We denote by $\linear{\cH_A}$ the set of linear operators 
and by $\1_A$ the identity operator in  the system $A$.
For a linear operator $X\in\linear{\cH_A}$,
we use $X^T$ to denote its transpose and $X^\dagger$ to denote its conjugate transpose.
The trace norm of $X$ is defined as $\norm{X}{1}:=\tr\sqrt{X^\dagger X}$.
The spectral norm of $X$ is defined as $\norm{X}{\infty}:= \sigma_{\max}(X)$,
where $\sigma_{\max}(X)$ is the largest singular value of $X$.
We denote by $\herm{\cH_A}$ the set of Hermitian operators and
$\pos{\cH_A}$ the set of positive semidefinite operators in the  system $A$.
Quantum states are positive semidefinite operators with unit trace. 
We write $X\geq0$ if and only if $X$ is positive semidefinite.

A linear map $\cN_{A\to B}$ is a linear map that transforms linear operators in system $A$ to linear operators in system $B$, i.e., $\cN_{A\to B}:\linear{\cH_A}\rightarrow \linear{\cH_B}$. 
 We use the calligraphic letters (e.g., $\cM$, $\cN$, $\cO$)
to represent the linear maps\footnote{Throughout this work, we are only interested in the linear maps for which the input and output quantum systems have the
same dimension.} and use $\id$ to represent  the identity map.
Let $\tr$ denote the trace function.
We say the linear map $\cN_{A\to A}$
is {trace-preserving} (TP) if $\tr[\cN(X)]=\tr[X]$ for arbitrary $X\in\linear{\cH_A}$,
is {trace non-increasing} (TN) if $\tr[\cN(X)]\leq\tr[X]$ 
for arbitrary $X\in\linear{\cH_A}$,
is {Hermitian-preserving} (HP) if $\cN(X)\in\herm{\cH_A}$ for arbitrary $X\in\herm{\cH_A}$,
is {positive} if $\cN(X)\in\pos{\cH_A}$ for arbitrary $X\in\pos{\cH_A}$,
and is {completely positive} (CP) if $\id_{R}\ox\cN$ is positive for
arbitrary reference system $R$.
We call a linear map CPTP if it is both completely positive and trace-preserving,
\emph{CPTN} if it is completely positive and trace non-increasing,
\emph{HPTP} if it is Hermitian- and trace-preserving.
Note that CPTP and CPTN maps are also known as quantum channels and quantum subchannels
in quantum information theory.

\begin{comment}
Let $\tilde{T}_u(d)$ be the sets of unitary channel. Let   $S_d=\{\cP_{\ket{\Psi}}|\ket{\Psi}\in \cH_d \}$ be the state preparation channel.  Let $\tilde{\cP}(d)$ be the convex combination of unitary and state preparation channels, that is,
 $\tilde{\cP}(d)=\{\sum p_i\cV_i|\cV_i\in\tilde{T}_u(d)\cup \cS_d\}$.
\end{comment}

Let $A$ be a $d$-dimensional quantum system and let $A'$ be isomorphic to $A$.
A maximally entangled state of rank $d$ in $A'A$ is defined
as $\ket{\Psi}_{A'A}:=\sum_{i=0}^{d-1}\ket{ii}/\sqrt{d}$, where
$\{\ket{i}\}_{i=0}^{d-1}$ forms an orthonormal basis of the system $A$.
We denote its unnormalized version by $\ket{\Gamma}:=\sqrt{d}\ket{\Psi}$.
Given a linear map $\cN_{A\to A}$, its Choi operator~\cite{choi1975completely} is defined as
\begin{align}\label{eq:Choi}
    J_{\cN}^{A'A}
:= (\id_{A'}\otimes \cN_{A\to A})(\proj{\Gamma}_{A'A})
 = \sum_{i,j=0}^{d-1} \ketbra{i}{j}_{A'}\otimes \cN_{A\to A}(\ketbra{i}{j}_A).
\end{align}
By the Choi-Jamio\l{}kowski isomorphism~\cite{choi1975completely,jamiolkowski1972linear},
$\cN$ is completely positive if and only if $J_{\cN}\geq 0$,
is trace-preserving if and only if $\tr_A{J_\cN}= \1_{A'}$, where
$\1_{A'}$ is the identity operator in $A'$,
is trace non-increasing if and only if $\tr_A{J_\cN}\leq\1_{A'}$, and
is Hermitian preserving if and only if $J_{\cN}^\dagger=J_{\cN}$
\footnote{We refer the interested readers to~\cite[Section 2.2]{watrous2018theory} and
the references therein for a detailed proof of these results.}.
Given the Choi operator $J_\cN$, the output of $\cN$ can be reconstructed
via
\begin{align}\label{eq:Choi-output}
    \cN(\rho) = \tr_{A'}[{(\rho^T\otimes I_A)}J_\cN],
\end{align}
where the transpose $T$ is with respect to the orthonormal basis defining $\ket{\Gamma}_{A'A}$.

We use $\bR$ to represent the real field and $\bC$ to represent the complex field.

%%%%%%%%%%%%%%%%%%%%%%%%%%%%%%%%%%%%%%%%%%%%%%%%%%%%%%%%%%%%%%%%%%%%%%%%%%%%%%%%%%%%%%%%%%%
\subsection{Invertible CPTP maps}

In this subsection, we formally define  the concept of invertible
CPTP maps~\cite{nayak2006invertible,shirokov2013reversibility,crowder2015linearization} and explore their properties.  
Invertible CPTP maps motivate our study of physical implementbility and will be investigated in detail  in Section~\ref{sec:Analytic expression}. Readers who are familiar with invertible maps may safely skip this section and come back when necessary.

Let $\cO_{A\to A}$ be a CPTP map in the system $A$. We say $\cO$ is \emph{invertible},
if there exists a linear map $\cN_{A\to A}$ in the system $A$ such that
\begin{align}\label{equ:inverse}
    \forall X\in\linear{\cH_A},\; \cN\circ\cO(X) = X.
\end{align}
{In this manuscript we only consider system $A$ of finite dimension, where the range of an invertible map $\mathcal{O}$ is the whole space, that is $\{\cO(X)\; |\; X\in \linear{\cH_A}\}=\linear{\cH_A}.$ }
That is, $\cN$ cancels the effect of $\cO$ and returns back the input state.
In this case, we adopt the convention $\cO^{-1}\equiv\cN$ and call
$\cO^{-1}$ the inverse map of $\cO$.
What's more, we say $\cO$ is \emph{strictly invertible}, if in addition the quantum
map $\cN$ is CPTP.  {Wigner's theorem}~\cite{wigner1960group} guarantees that
a CPTP map $\cO_{A\to A}$ is strictly invertible if and only if it is a unitary channel,
i.e., $\cO(\cdot)=U(\cdot)U^\dagger$ for some unitary $U$. On the other hand,
the qudit depolarizing channel
is invertible but not strictly invertible since its inverse linear map is not completely positive
(cf. Lemma~\ref{lemma:depolarizing}).
We emphasize that not all CPTP maps are invertible;
the constant quantum channel~\cite{cooney2016strong,berta2018amortized,Fang2018,Wang2017a,takagi2020application},
which maps all inputs into some fixed quantum state, is a prominent counterexample.
%In Appendix~\ref{appx:invertible qubit CPTP},we carry out a detailed analysis of the invertible qubit CPTP maps. The obtained results reveal that the qubit case already exhibits a  rather rich structure. 

In the following, we show that for the system $A$ of finite dimension, if a CPTP map $\cO_{A\to A}$ \emph{is} invertible,
its inverse map $\cO^{-1}$ must necessarily be both 
Hermitian- and trace-preserving (HPTP).

\begin{property}\label{pro:hermitian}
Let $\cO_{A\to A}$ be an invertible CPTP map. 
The following statements hold:
\begin{enumerate}
  \item $\cO^{-1}$ is Hermitian-preserving; and
  \item $\cO^{-1}$ is trace-preserving.
\end{enumerate}
\end{property}
\begin{proof}
First notice that for arbitrary $X\in\linear{\cH_A}$,
\begin{align}\label{eq:gLVWPlFmP}
\cO^{-1}\circ\cO(\rho) = \rho
\Leftrightarrow \cO\circ\cO^{-1}\circ\cO(\rho) = \cO(\rho)
\Leftrightarrow \cO\circ\cO^{-1}\left(\cO(\rho)\right) = \cO(\rho)
\Rightarrow\cO\circ\cO^{-1} = \id_{\cO},
\end{align}
where $\id_{\cO}$ is the identity map on the range of $\cO$.

Now we show the Hermitian-preserving property. 
Assume the system $A$ is $d$-dimensional. Then $\linear{\cH_A}$ is a $d^2$-dimensional linear space over the complex field $\bC$ and $\herm{\cH_A}$ is a $d^2$-dimensional linear space over the real field $\bR$.
As so, there exist linearly independent operators $O_1,...,O_{d^2}\in \herm{\cH_A}$\footnote{For example, write ${E_{ij}\in\bR^{d\times d}}$ as the matrix that $E_{ij}[k,l]$ equals $1$ for $k,l=i,j$ and $0$ otherwise. We can choose the basis as $\{E_{kk}\}_{k\in[d]}; \{E_{kl}-E_{lk}\}_{ k<l; k,l\in[d]}; \{iE_{kl}-iE_{lk}\}_{ k<l; k,l\in[d]}$.} such that for every $O\in \herm{\cH_A}$,
\begin{align}
    O=\sum_{i=1}^{d^2} c_i O_i, c_i\in  \bR.
\end{align}
Since $\cO$ is a quantum channel and Hermitian-preserving, we have $\cO(O_i)\in \herm{H_A}$.
Since $\cO$ is linear and invertible, $\cO(O_1),...,\cO(O_{d^2})$ must be linearly independent over the field $\bR$. Otherwise, there exists $O\neq\bm 0$ for which $\cO(O)=\bm 0$, implying that $\cO$ is not invertible since by linearity we already have {$\cO^{-1}(\bm 0)=\bm 0$}. To conclude, $\{\cO(O_1),...,\cO(O_{d^2})\}$ is a linear independent set in $\herm{H_A}$ over the field $\bR$
of size $d^2$, thus it forms a basis for $\herm{\cH_A}$. Then for any $O\in \herm{\cH_A}$, there exists $\{a_i\in \bR\}_i$ such that
\begin{align}
    O&=\sum_{i=1}^{d^2} a_i\cO(O_i),\\
    &=\cO\left(\sum_{i=1}^{d^2} a_iO_i\right),
\end{align}
where the first equality follows since $\{\cO(O_i)\}_i$ forms a basis, and the second equality follows since $\cO$ is linear. Now it holds that for any $O\in \herm{\cH_A}$,
\begin{align}
    \cO^{-1}(O)
= \cO^{-1}\circ\cO\left(\sum_{i=1}^{d^2} a_iO_i\right)
= \sum_{i=1}^{d^2} a_iO_i\in\herm{\cH_A},
\end{align}
implying that $\cO^{-1}$ is Hermitian-preserving.

The trace-preserving property is easy to check. {Notice that in this work we only consider system $A$ of finite dimension, where the range of an invertible map $\mathcal{O}$ is the whole space, i.e., $\{\cO(X)\; |\; X\in \linear{\cH_A}\}=\linear{\cH_A}$.}
For arbitrary $X\in\linear{\cH_A}$, we have
\begin{align}
    \tr\left[\cO^{-1}(X)\right]
=   \tr\left[\cO\circ\cO^{-1}(X)\right]
=   \tr\left[X\right],
\end{align}
where the first equality follows from trace-preserving property of $\cO$
and the second equality follows from the definition of inverse map.
\end{proof}

%%%%%%%%%%%%%%%%%%%%%%%%%%%%%%%%%%%%%%%%%%%%%%%%%%%%%%%%%%%%%%%%%%%%%%%%%%%%%%%%%%%%%%%%%%%%%
%%%%%%%%%%%%%%%%%%%%%%%%%%%%%%%%%%%%%%%%%%%%%%%%%%%%%%%%%%%%%%%%%%%%%%%%%%%%%%%%%%%%%%%%%%%%%
%%%%%%%%%%%%%%%%%%%%%%%%%%%%%%%%%%%%%%%%%%%%%%%%%%%%%%%%%%%%%%%%%%%%%%%%%%%%%%%%%%%%%%%%%%%%%
\section{Physical implementability of linear maps}\label{sec:physical implementability}

In this section, we introduce the physical implementability measure to characterize
how well a linear map can be physically approximated.  We investigate its various properties,
derive bounds in terms of its Choi operator's trace norm, and deliver analysis for particular examples of interest.
%and obtain analytic expressions for some linear maps of practical interest.

\subsection{Definition}

Completely positive and trace-preserving (CPTP)  maps mathematically characterize physically
implementable quantum operations in a given quantum system.
Nevertheless, general linear maps, such as positive but not completely positive maps,
which are impossible to be physically implemented, are also fundamental ingredients
from theoretical and practical perspectives. For example, in the error mitigation task, one might wish to implement the inverse map of the noise, which may not be completely positive thus is non-physical.
It leads naturally to the problem of physically approximating these `non-physical' linear maps. To this end, we introduce the physical implementability measure
to characterize that how well a given quantum linear map can be \JJQ{physically} implemented. 
Here we interpret ``physically approximating a linear map"  as decomposing the given map into linear combination of CPTPs, and quantify the hardness of approximation by the $l_1$-norm of the decomposition coefficients. This interpretation is inspired by the error mitigation task and \JJQ{its} operational meaning will be further discussed in Section \ref{sec:error mitigation}. 

Formally, given an HPTP map $\cN$ in $\cH$, the \emph{physical implementability} of $\cN$ is defined as
\begin{align}\label{eq:implementability}
    \ce{\cN}:=\log\min
    \left\{ \sum_{\alpha} | \eta_{\alpha} |
        \sbar \cN= \sum_{\alpha}  \eta_{\alpha} \cO_{\alpha},\;
        \cO_{\alpha} \text{~is CPTP},\; \eta_{\alpha}\in \bR
    \right\},
\end{align}
where logarithms are in base $2$ throughout this paper.
Intuitively, we decompose the map $\cN$
as a linear combination of CPTP maps that are physically implementable in a quantum system;
negative terms, that is $\eta_\alpha<0$, must be introduced and they quantify
the fundamental limit of physical implementability of $\cN$.
In the following Lemma~\ref{lem:N2} we show 
that $\cN$ can always be decomposed into a 
linear combination of two quantum channels with carefully chosen coefficients.  Lemma \ref{lem:N2} and SDP (\ref{eq:r-SDP-Prime}) show that the  value of (\ref{eq:implementability}) is  the minimization of a convex function over a closed convex set, thus \JJQ{we can} write $\min$ instead of $\inf$ in (\ref{eq:implementability}). 
Interestingly, the physical implementability measure finds its operational meaning
in error mitigation tasks, as we will argue in the next section.

The measure $\nu$ bears nice properties. First of all, it is efficiently computable via
semidefinite programs.
It satisfies the desirable additivity property with respect to tensor products.
This property ensures that the parallel application of linear maps cannot
make the physical implementation `easier' compared to implementing these linear maps individually.
It also satisfies the monotonicity property with respect to quantum superchannels.
We give upper and lower bounds on the physical implementability in terms of its Choi operator's trace norm,
connecting this measure to the widely studied operator norm in the literature.
These bounds are tight in the sense that there exist HPTP maps for
which the bounds are saturated. What's more, we are able to
derive analytical expressions of this measure for some inverse maps of practically interesting CPTP maps.

\subsection{Semidefinite programs}

In this section, we propose a semidefinite program (SDP) calculating $\ce{\cN}$ for arbitrary HPTP maps $\cN$. We remind that under mild regularity assumptions, SDP 
can be efficiently solved by interior point method~\cite{boyd2004convex} with runtime polynomial
in $d$, where $d$ is the dimension of the quantum system $\cH$.
This algorithm is designed by noticing that the optimal
decomposition~\eqref{eq:implementability} can always be obtained by just two quantum channels.

In the following, we prove that HPTP maps can always be decomposed into a linear combination of
two quantum channels with carefully chosen coefficients. It is worth noting
that, the construction in Lemma \ref{lem:N2} does not yield the optimal physical
implementability.
However, in Theorem \ref{thm:N2} we show the optimality can be obtained by another two CPTP maps.

\begin{lemma}\label{lem:N2}
  Let $\cN$ be an HPTP map, then there exist two real numbers $\eta_1\geq 0,\eta_2\geq 0$
  and two CPTP maps $\cO_1,\cO_2$ such that
  \begin{align}
    \cN=\eta_1\cO_1-\eta_2\cO_2.
  \end{align}
\end{lemma}
\begin{proof}
By the Choi-Jamio\l{}kowski isomorphism, it's equivalent to prove
that there exist two real numbers $\eta_1$, $\eta_2$,
and two positive semidefinite operators $J_1$ and $J_2$, such that $J_i\geq 0$,
$\tr_B J_i = \1_A$, $i=1,2$, and
\begin{align}\label{eq:temp}
J_\cN=\eta_1 J_1-\eta_2 J_2.
\end{align}
From solution of \eqref{eq:temp}
one can construct desired $\cO_1,\cO_2$ via  \eqref{eq:Choi-output}.

Write $d_A,d_{B}$ as the dimensions of system $A,B$ respectively. Recall that $\tr_B J_\cN= \1_A$, it suffices to let
\begin{align}
  & \eta_1=(\norm{J_\cN}{1}+1)\cdot d_{B}, J_1=I_{AB}/d_{B}, \\
  & \eta_2=(\norm{J_\cN}{1}+1)\cdot d_B-1, J_2= (\eta_1J_1-J_\cN)/\eta_2
  \end{align}
\end{proof}

\begin{theorem}\label{thm:N2}
Let $\cN$ be an HPTP map. It holds that
\begin{align}\label{eq:N2}
   \ce{\cN}=\log\min\left\{ \eta_1+\eta_2  \sbar \cN=\eta_1\cO_1-\eta_2\cO_2;\;
              \eta_i\geq 0,\; \cO_i\rm{~is~CPTP}\right\}.
\end{align}
\end{theorem}
\begin{proof}
From Lemma \ref{lem:N2} we know $\cN$ can always be written as linear combination of two CPTPs.
This implies that $\ce{\cN}$ is finite. We prove
that $\ce{\cN}$
can always be obtained by a linear combination of two CPTPs. Suppose on the contrary
that $\ce{\cN}$ is achieved by
\begin{align}\label{eq:whtJbaZQN}
\cN &= \sum_{\alpha\in\cK}\eta_\alpha' \cO_\alpha',\;
            \eta_\alpha'\in\cR,\; \cO_\alpha'~~\text{is CPTP},
\end{align}
where $\vert\cK\vert\geq3$ and might be infinite (in which case the sum shall be replaced by the
integral in~\eqref{eq:whtJbaZQN}). That is, $\ce{\cN}=\sum_{\alpha\in\cK} |\eta_{\alpha}'|$.

Noticing the set of CPTPs is convex, we divide the set of quantum channels
in~\eqref{eq:whtJbaZQN} into two subgroups according to the sign of their coefficients.
More precisely, set
\begin{align}
\eta_1 &:= \sum_{\alpha: \eta_\alpha'\geq 0}\eta_\alpha',\\
\cO_1  &:= {\eta_1}\sum_{\alpha:\eta_\alpha'\geq 0} \frac{\eta_\alpha'}{\eta_1} \cO_\alpha',\\
\eta_2 &:= \sum_{\alpha: \eta_\alpha'<0}|\eta_\alpha'|,\\
\cO_2  &:= {\eta_2}\sum_{\alpha:\eta_\alpha'< 0} \frac{|\eta_\alpha'|}{\eta_2} \cO_\alpha',
\end{align}
we can easily check from~\eqref{eq:whtJbaZQN} that
\begin{align}
    \cN = \eta_1\cN_1 - \eta_2\cN_2.
\end{align}
What's more, the above decomposition yields
\begin{align}
  \eta_1 + \eta_2 = \sum_{\alpha'}|\eta_{\alpha}'| = \ce{\cN}.
\end{align}
\end{proof}

\begin{remark}
Note that a similar conclusion to Theorem~\ref{thm:N2} has been obtained in
\cite{takagi2020optimal} . The proof here is  similar.
\end{remark}

\vb

An advantage of Theorem~\ref{thm:N2} is that it leads to a semidefinite program
characterization of the measure $\ce{\cN}$ (more concretely, $2^{\ce{\cN}}$),
in terms of its the Choi operator $J_\cN$:
\begin{subequations}\label{eq:r-SDP-Prime}
\begin{align}
\text{\bf Primal:}\quad
   2^{\ce{\cN}} =     \min&\; p_1+p_2 \\
         \text{s.t.}&\; J_{\cN} = J_1-J_2\label{eq:r-SDP-Prime-1} \\
                    &\; \tr_B J_1 = p_1\1_A\label{eq:r-SDP-Prime-2} \\
                    &\; \tr_B J_2 = p_2\1_A\label{eq:r-SDP-Prime-3} \\
                    &\; J_1 \geq 0, J_2 \geq 0
\end{align}
\end{subequations}
Throughout this paper, `s.t.' is short for `subject to'.
 Correspondingly, the dual SDP is given by
(the proof can be found in Appendix~\ref{appx:proof-of-the-dual})
\begin{subequations}\label{eq:r-SDP-Dual}
\begin{align}
\text{\bf Dual:}\quad
   2^{\ce{\cN}} = \max &\; \tr\left[M_{AB} J_{\cN}\right] \\
            \text{s.t.}&\; \tr N_A = 1 \\
                         &\; \tr K_A = 1 \\
                       &\; M_{AB} + N_A\otimes \1_B \geq 0 \\
                       &\; - M_{AB} + K_A\otimes \1_B \geq 0
\end{align}
\end{subequations}
where the maximization ranges over all Hermitian operators $M_{AB}$ 
in the joint system $AB$,
and Hermitian operators $N_A, K_A$ in the system $A$. One can check that the primal SDP satisfies the Slater condition, thus the strong duality holds.
These primal and dual semidefinite programs are
useful in proving the properties of $\ce{\cN}$.

Besides, we show that relaxing the condition in~\eqref{eq:N2} from CPTPs to
CPTN (completely positive and trace non-increasing) maps leads to the same $\ce{\cN}$.

\begin{lemma}\label{lem:N2-sub}
Let $\cN$ be an HPTP map. It holds that
\begin{align}\label{eq:N2-sub}
   \ce{\cN}=\log\min\left\{ \eta_1+\eta_2  \sbar \cN=\eta_1\cO_1-\eta_2\cO_2;\;
              \eta_i\geq 0,\; \cO_i \rm{~is~CPTN}\right\}.
\end{align}
\end{lemma}
\begin{proof}
Define the following optimization problem:
\begin{subequations}\label{eq:r-SDP-Prime-alt}
\begin{align}
   2^{\omega(\cN)} =     \min&\; p_1+p_2 \\
         \text{s.t.}&\; J_{\cN} = J_1-J_2\label{eq:r-SDP-Prime-1-alt} \\
                    &\; \tr_B J_1 \leq p_1\1_A\label{eq:r-SDP-Prime-2-alt} \\
                    &\; \tr_B J_2 \leq p_2\1_A\label{eq:r-SDP-Prime-3-alt} \\
                    &\; J_1 \geq 0, J_2 \geq 0
\end{align}
\end{subequations}
It is equivalent to show that $\ce{\cN} = \omega(\cN)$. Note that the only difference
between the programs~\eqref{eq:r-SDP-Prime} and~\eqref{eq:r-SDP-Prime-alt} lies in that
we relax the equality conditions in~\eqref{eq:r-SDP-Prime-2} and~\eqref{eq:r-SDP-Prime-3}
to inequality conditions.

By definition it holds that $\ce{\cN} \geq \omega(\cN)$ since the constraints are relaxed.
We prove $\ce{\cN} \leq \omega(\cN)$ by showing that we are able to construct
a feasible solution of~\eqref{eq:r-SDP-Prime} from
any feasible solution of~\eqref{eq:r-SDP-Prime-alt} with the same performance.
Assume $(p_1,p_2,J_1,J_2)$ achieves~\eqref{eq:r-SDP-Prime-alt}. That is,
\begin{align}\label{eq:nQiIEwIFcSHR}
2^{\omega(\cN)} = p_1 + p_2, \quad
J_{\cN} = J_1 - J_2,\quad
\tr_BJ_1 \leq p_1\1_A,\quad
\tr_BJ_2 \leq p_2\1_A.
\end{align}
Define the following operators
\begin{align}
\wt{J}_1 &:= J_1 + \left[\frac{p_1+p_2}{2}\1_A - \frac{\tr_BJ_1+\tr_BJ_2}{2}\right]\ox\1_B/d_B, \\
\wt{J}_2 &:= J_2 + \left[\frac{p_1+p_2}{2}\1_A - \frac{\tr_BJ_1+\tr_BJ_2}{2}\right]\ox\1_B/d_B.
\end{align}
By~\eqref{eq:nQiIEwIFcSHR},
we have $\wt{J}_1,\wt{J}_2\geq0$ and $\wt{J}_1 - \wt{J}_2 = J_1 - J_2 = J_{\cN}$.
What's more,
\begin{align}
    \tr_B\wt{J}_1
&= \tr_B J_1 + \frac{p_1+p_2}{2}\1_A - \frac{\tr_BJ_1+\tr_BJ_2}{2} \\
&= \frac{p_1+p_2}{2}\1_A + \frac{\tr_BJ_1-\tr_BJ_2}{2} \\
&= \frac{p_1+p_2+1}{2}\1_A, \\
    \tr_B\wt{J}_2
&= \tr_B J_2 + \frac{p_1+p_2}{2}\1_A - \frac{\tr_BJ_1+\tr_BJ_2}{2} \\
&= \frac{p_1+p_2}{2}\1_A - \frac{\tr_BJ_1-\tr_BJ_2}{2} \\
&= \frac{p_1+p_2-1}{2}\1_A,
\end{align}
where we used the fact that $\tr_BJ_1-\tr_BJ_2 = \tr_BJ_{\cN} = \1_A$.
Since $p_1+p_2\geq1$ by definition, $\wt{J}_2$ is well defined.
That is to say,
$((p_1+p_2+1)/2,(p_1+p_2-1)/2,\wt{J}_1,\wt{J}_2)$ is a feasible solution for~\eqref{eq:r-SDP-Prime}
and thus
\begin{align}
  2^{\ce{\cN}} \leq \frac{p_1+p_2+1}{2} + \frac{p_1+p_2-1}{2} = p_1+p_2 = 2^{\omega(\cN)}.
\end{align}
\end{proof}

\subsection{Properties}

In this section, we prove several interesting properties of $\ce{\cN}$
-- faithfulness, additivity w.r.t. (with respect to) tensor product of linear maps,
subadditivity w.r.t. linear map composition, unitary channel invariance,
and monotonicity under quantum superchannels -- via exploring its semidefinite program characterizations.

The faithfulness property ensures
a linear map possesses zero physical implementability measure
if and only if it is physically implementable, i.e., it is a CPTP map.
As so, positive physical implementability measure indicates that
the corresponding map is not physically implementable.

\begin{lemma}[Faithfulness]\label{lemma:faithfulness}
Let $\cN_{A\to A}$ be an HPTP map. $\ce{\cN}=0$ if and only if $\cN$ is CPTP.
\end{lemma}
\begin{proof}
The `if' part follows directly by definition.
To show the `only if' part, recall
the alternative characterization of $\ce{\cdot}$ in Theorem~\ref{thm:N2}.
Since $\ce{\cN}=0$, there must exist a channel ensemble $\{(\eta_1,\cO_1),(\eta_2,\cO_2)\}$ such that
$\cN = \eta_1\cO_1 - \eta_2\cO_2$, $\eta_1,\eta_2\geq0$, and
$\eta_1+\eta_2=1$. Since $\cN$ is trace-preserving,
it holds that $\eta_1-\eta_2=1$.
These conditions yield $\eta_1=1, \eta_2=0$ 
and thus $\cN=\cO_1$, implying that $\cN$ is completely positive.
\end{proof}

\begin{theorem}[Additivity w.r.t. tensor product]\label{thm:additivity}
Let $\cM_{A_1\to B_1}$ and $\cN_{A_2\to B_2}$ be two HPTP maps. It holds that
\begin{align}
  \ce{\cM\ox\cN} = \ce{\cM} + \ce{\cN}.
\end{align}
\end{theorem}
\begin{proof}
It is equivalent to show that
\begin{align}\label{eq:multiplicativity}
2^{\ce{\cM\ox\cN}} = 2^{\ce{\cM}} \cdot 2^{\ce{\cN}}.
\end{align}
We prove~\eqref{eq:multiplicativity} by exploring the primal and dual SDP characterizations
of $\nu$ in~\eqref{eq:r-SDP-Prime} and~\eqref{eq:r-SDP-Dual}.

``$\leq$'': Assume the tetrad $(p_1,p_2,P_1,P_2)$ achieves $2^{\ce{\cM}}$ and
the tetrad $(q_1,q_2,Q_1,Q_2)$ achieves $2^{\ce{\cN}}$ w.r.t.~\eqref{eq:r-SDP-Prime}.
That is,
\begin{align}
    2^{\ce{\cM}} &= p_1 + p_2,\; J_\cM = P_1 - P_2,\;
      \tr_{B_1} P_1 = p_1\1_{A_1},\; \tr_{B_1} P_2 = p_2\1_{A_1}, \\
    2^{\ce{\cN}} &= q_1 + q_2,\; J_\cN = Q_1 - Q_2,\;
      \tr_{B_2} Q_1 = q_1\1_{A_2},\; \tr_{B_2} Q_2 = q_2\1_{A_2}.
\end{align}
What's more, since $\cM$ is trace-preserving, we have $p_1-p_2=1$. Similarly, $q_1-q_2=1$.
Notice that
\begin{align}
  J_{\cM\ox\cN}
&= J_{\cM} \ox J_{\cN} \\
&= (P_1 - P_2) \ox (Q_1 - Q_2) \\
&= \left(P_1\ox Q_1 + P_2 \ox Q_2\right) - (P_1\ox Q_2 + P_2\ox Q_1).
\end{align}
This yields a feasible decomposition of $J_{\cM\ox\cN}$. Further more, since
\begin{align}
    \tr_{B_1B_2}\left[P_1\ox Q_1 + P_2 \ox Q_2\right]
&=    \tr_{B_1} P_1 \ox \tr_{B_2} Q_1
    + \tr_{B_1} P_2 \ox \tr_{B_2} Q_2 \\
&= \left(p_1q_1 + p_2q_2\right) \1_{A_1A_2}, \\
    \tr_{B_1B_2}\left[P_1\ox Q_2 + P_2 \ox Q_1\right]
&=    \tr_{B_1} P_1 \ox \tr_{B_2} Q_2
    + \tr_{B_1} P_2 \ox \tr_{B_2} Q_1 \\
&= \left(p_1q_2 + p_2q_1\right) \1_{A_1A_2},
\end{align}
it holds that
\begin{align}
      2^{\ce{\cM\ox\cN}}
\leq  p_1q_1 + p_2q_2 + p_1q_2 + p_2q_1
= p_1(q_1 + q_2) + p_2(q_1 + q_2)
= 2^{\ce{\cM}} \cdot 2^{\ce{\cN}}.
\end{align}

``$\geq$'': Assume the triple $(M_1,N_1,K_1)$ achieves $2^{\ce{\cM}}$
and the triple $(M_2,N_2,K_2)$ achieves $2^{\ce{\cN}}$ w.r.t.~\eqref{eq:r-SDP-Dual}.
That is,
\begin{align}
    2^{\ce{\cM}} &= \tr[M_1J_\cM],\; 
      \tr N_1 = 1,\;
      \tr K_1 = 1,\; 
      M_1 + N_1\ox\1_{B_1} \geq 0,\; 
      - M_1 + K_1\ox\1_{B_1} \geq 0,\label{eq:assumption1} \\
    2^{\ce{\cN}} &= \tr[M_2J_\cN],\; 
      \tr N_2 = 1,\;
      \tr K_2 = 1,\;
      M_2 + N_2\ox\1_{B_2} \geq 0,\; 
      - M_2 + K_2\ox\1_{B_2} \geq 0.\label{eq:assumption2}
\end{align}
As a direct consequence, we obtain
\begin{subequations}\label{eq:trace-conditions}
\begin{align}
   1 &= \tr[N_1]=\tr[N_1^T] = \tr[\cM(N_1^T)]
                            = \tr_{AB}[(N_1\ox \1_{B_1})J_\cM], \\
   1 &= \tr[K_1]=\tr[K_1^T] = \tr[\cM(K_1^T)]
                            = \tr_{AB}[(K_1\ox \1_{B_1})J_\cM], \\
   1 &= \tr[N_2]=\tr[N_2^T] = \tr[\cN(N_2^T)] 
                            = \tr_{AB}[(N_2\ox \1_{B_2})J_\cN], \\
   1 &= \tr[K_2]=\tr[K_2^T] = \tr[\cN(K_2^T)] 
                            = \tr_{AB}[(K_2\ox \1_{B_2})J_\cN],
\end{align}
\end{subequations}
where we have used~\eqref{eq:Choi-output}
and the fact that $T$, \JJQ{$\cM$, and $\cN$} are all trace-preserving to derive the above relations.

Define the following operators:
\begin{align}
    \wt{M}_{A_1A_2B_1B_2} &:= M_1\ox M_2 
        + \frac{M_1\ox(N_2 - K_2) + (N_1 - K_1)\ox M_2}{2}, \\
    \wt{N}_{A_1A_2} &:= \frac{N_1\ox N_2 + K_1\ox K_2}{2}, \\
    \wt{K}_{A_1A_2} &:= \frac{N_1\ox K_2 + K_1\ox N_2}{2}.
\end{align}
In the following, we show that the triple $(\wt{M}, \wt{N}, \wt{K})$
is a feasible solution to $2^{\ce{\cM\ox\cN}}$ since it satisfies all the
constraints in the dual SDP~\eqref{eq:r-SDP-Dual}.
First of all, we have
\begin{align}
  \tr[\wt{N}_{A_1A_2}]
&= \frac{1}{2}\left(\tr[N_1]\tr[N_2] + \tr[K_1]\tr[K_2]\right) = 1, \\
  \tr[\wt{K}_{A_1A_2}]
&= \frac{1}{2}\left(\tr[N_1]\tr[K_2] + \tr[K_1]\tr[N_2]\right) = 1,
\end{align}
where we use the constraints that $\tr N_i= 1$, $\tr K_i=1$ for $i=1,2$. Second, we have
\begin{align}
 &\;    2\left(\wt{M}_{A_1A_2B_1B_2} + \wt{N}_{A_1A_2}\ox\1_{B_1B_2}\right) \\
=&\; 2M_1\ox M_2 + M_1\ox(N_2 - K_2) + (N_1 - K_1)\ox M_2 + N_1\ox N_2 \ox\1_{B_1B_2} + K_1\ox K_2 \ox\1_{B_1B_2} \\
=&\; \left(M_1 + N_1\ox\1_{B_1}\right)\ox\left(M_2 + N_2\ox\1_{B_2}\right)
    + \left(- M_1 + K_1\ox\1_{B_1}\right)\ox\left(- M_2 + K_2\ox\1_{B_2}\right) \\
\geq&\; 0,
\end{align}
where the inequality follows from Eqs.~\eqref{eq:assumption1} and~\eqref{eq:assumption2}.
This gives
\begin{align}
    \wt{M}_{A_1A_2B_1B_2} + \wt{N}_{A_1A_2}\ox\1_{B_1B_2}\geq0.
\end{align}
Similarly, we have
\begin{align}
 &\;    2\left(- \wt{M}_{A_1A_2B_1B_2} + \wt{K}_{A_1A_2}\ox\1_{B_1B_2}\right) \\
=&\; - 2M_1\ox M_2 - M_1\ox(N_2 - K_2) - (N_1 - K_1)\ox M_2 + N_1\ox K_2 \ox\1_{B_1B_2} + K_1\ox N_2  \ox\1_{B_1B_2}\\
=&\; \left(M_1 + N_1\ox\1_{B_1}\right)\ox\left(-M_2 + K_2\ox\1_{B_2}\right)
    + \left(- M_1 + K_1\ox\1_{B_1}\right)\ox\left(M_2 + N_2\ox\1_{B_2}\right) \\
\geq&\; 0,
\end{align}
yielding
\begin{align}
    - \wt{M}_{A_1A_2B_1B_2} + \wt{K}_{A_1A_2}\ox\1_{B_1B_2}\geq0.
\end{align}
This concludes that the triple $(\wt{M}, \wt{N}, \wt{K})$ is a feasible solution.

Since the triple $(\wt{M}, \wt{N}, \wt{K})$ is a feasible solution, we conclude that
\begin{align}
    2^{\ce{\cM\ox\cN}}
&\geq \tr[\wt{M}_{A_1A_2B_1B_2}J_{\cM\ox\cN}] \\
&=  \tr\left[M_1 J_\cM\right]\tr\left[M_2 J_\cN\right] \\
&\qquad    + \frac{\tr\left[M_1 J_\cM\right]\tr\left[(N_2-K_2)J_{\cN}\right]
      + \tr\left[(N_1-K_1) J_\cM\right]\tr\left[M_2 J_{\cN}\right]}{2} \\
&=  \tr\left[M_1 J_\cM\right]\tr\left[M_2 J_\cN\right] \\
&=  2^{\ce{\cM}}\cdot 2^{\ce{\cM}},
\end{align}
where the second equality follows from~\eqref{eq:trace-conditions}
and the last equality follows from~\eqref{eq:assumption1} and~\eqref{eq:assumption2}.
We are done.
\end{proof}

\begin{theorem}[Subadditivity w.r.t. composition]\label{thm:composition-subadditivity}
Let $\cM_{A_2\to A_3}$ and $\cN_{A_1\to A_2}$ be two HPTP maps.
It holds that
\begin{align}\label{eq:composition-submultiplicativity}
  \ce{\cM\circ\cN} \leq \ce{\cM} + \ce{\cN}.
\end{align}
\end{theorem}
\begin{proof}
It is equivalent to show that
\begin{align}\label{eq:composition-submultiplicativity-2}
2^{\ce{\cM\circ\cN}} \leq 2^{\ce{\cM}} \cdot 2^{\ce{\cN}}.
\end{align}
We prove~\eqref{eq:composition-submultiplicativity-2} by exploring
the alternative characterization of $\ce{\cdot}$ proved in Theorem~\ref{thm:N2}.
Assume the ensemble $\{(p_1,\cP_1),(p_2,\cP_2)\}$ achieves $2^{\ce{\cM}}$
and the ensemble $\{(q_1,\cQ_1),(q_2,\cQ_2)\}$ achieves $2^{\ce{\cN}}$.
That is,
\begin{align}
2^{\ce{\cM}} &= p_1+p_2,\quad  \cM = p_1\cP_1 - p_2\cP_2,\\
2^{\ce{\cN}} &= q_1+q_2,\quad \cN = q_1\cQ_1 - q_2\cQ_2.
\end{align}
Notice that
\begin{align}
  \cM\circ\cN
&= \left(p_1\cP_1 - p_2\cP_2\right)\circ\left(q_1\cQ_1 - q_2\cQ_2\right) \\
&= \left[p_1q_1\cP_1\circ\cQ_1 + p_2q_2\cP_2\circ\cQ_2\right]
 - \left[p_1q_2\cP_1\circ\cQ_2 + p_2q_1\cP_2\circ\cQ_1\right] \\
&\equiv r_1\cO_1 - r_2\cO_2,\label{eq:KXPtaRgLSo}
\end{align}
where
\begin{align}
    r_1 &:= p_1q_1 + p_2q_2, \\
    r_2 &:= p_1q_2 + p_2q_1, \\
    \cO_1 &:= \frac{p_1q_1}{r_1}\cP_1\circ\cQ_1 + \frac{p_2q_2}{r_1}\cP_2\circ\cQ_2, \\
    \cO_2 &:= \frac{p_1q_2}{r_2}\cP_1\circ\cQ_2 + \frac{p_2q_1}{r_2}\cP_2\circ\cQ_1.
\end{align}
Since the set of quantum channels is closed under convex combination,
the above defined $\cO_1$ and $\cO_2$ are valid quantum channels.
Then~\eqref{eq:KXPtaRgLSo} yields a feasible decomposition of $\cM\circ\cN$, resulting
\begin{align}
  2^{\ce{\cM\circ\cN}} \leq r_1 + r_2
= (p_1+p_2)(q_1+q_2)
= 2^{\ce{\cM}} \cdot 2^{\ce{\cN}}.
\end{align}
\end{proof}

As a direct corollary of the subadditivity property w.r.t. the composition operation,
we find that the implementability measure is invariant under unitary
quantum channels for both pre- and post-processing.

\begin{corollary}[Unitary channel invariance]\label{lemma:unitary invariance}
Let $\cN_{A\to A}$ be an HPTP map. For arbitrary channels $\cU(\cdot):=U(\cdot)U^\dagger$
and $\cV(\cdot):=V(\cdot)V^\dagger$, where $U$ and $V$ are unitaries, it holds that
\begin{align}
  \ce{\cU\circ\cN\circ\cV} = \ce{\cN}.
\end{align}
\end{corollary}
\begin{proof}
First of all, we have $\ce{\cU}=0$ for arbitrary unitary channel $\cU$ by
Lemma~\ref{lemma:faithfulness}.
Eq.~\eqref{eq:composition-submultiplicativity} yields
\begin{align}\label{eq:kvtPquz-1}
    \ce{\cU\circ\cN\circ\cV}
\leq \ce{\cU} + \ce{\cN} + \ce{\cV} = \ce{\cN}.
\end{align}
For the unitary channel $\cU(\cdot):=U(\cdot)U^\dagger$,
set $\cU^\dagger(\cdot):=U^\dagger(\cdot)U$. One can check that
$\cU^\dagger\circ\cU=\cU\circ\cU^\dagger=\id$.
That is, $\cU^\dagger$ is the inverse channel of $\cU$.
Notice that
\begin{align}
    \cU^\dagger \circ \left(\cU\circ\cN\circ\cV\right)\circ\cV^\dagger = \cN.
\end{align}
Applying Eq.~\eqref{eq:composition-submultiplicativity} to the above equality leads to
\begin{align}\label{eq:kvtPquz-2}
    \ce{\cN} = \ce{\cU^\dagger \circ \left(\cU\circ\cN\circ\cV\right)\circ\cV^\dagger}
    \leq \ce{\cU^\dagger}+ \ce{\cU\circ\cN\circ\cV} + \ce{\cV^\dagger}
  = \ce{\cU\circ\cN\circ\cV}.
\end{align}
Eqs.~\eqref{eq:kvtPquz-1} and~\eqref{eq:kvtPquz-2} together conclude the proof.
\end{proof}

\vb
From the resource theory perspective~\cite{chitambar2019quantum},
the set of linear maps under investigation is HPTP maps
(Hermitian- and trace-preserving maps),
while the set of \emph{free maps} is CPTP maps
(completely positive and trace-preserving maps),
because they can be perfectly physically implemented.
Quantum superchannels transform CPTP maps to CPTP maps~\cite{chiribella2008transforming},
thus they serve as a natural candidate of \emph{free supermaps} since they do not
incur physical implementability when operating on a CPTP map.
We show in the following the \emph{monotonicity} property, which
states that a superchannel can never render a linear map
more difficult to be physically implemented. 

\begin{theorem}[Monotonicity]\label{thm:Monotonicity}
Let $\cN_{A\to A}$ be an HPTP map and let $\Theta$ be a superchannel.
It holds that
\begin{align}
    \ce{\Theta(\cN)} \leq \ce{\cN}.
\end{align}
\end{theorem}
\begin{proof}
Since $\Theta$ is a superchannel, there exist a Hilbert space $\cH_E$ wih $d_E\leq d_A$
and two CPTP maps $\cP^{\text{pre}}_{A\to AE}$, $\cP^{\text{post}}_{AE\to A}$
such that~\cite{chiribella2008transforming}
\begin{align}\label{eq:superchannel}
\Theta(\cN)
= \cP^{\text{post}}_{AE\to A}\circ\left(\cN_{A\to A}\ox \id_{E}\right)
            \circ \cP^{\text{pre}}_{A\to AE}.
\end{align}
See Fig.~\ref{fig:superchannel} for illustration.
Assume the decomposition $\cN= \sum_{\alpha}\eta_{\alpha} \cO_{\alpha}$
achieves $\ce{\cN}$ w.r.t.~\eqref{eq:implementability},
i.e., $\ce{\cN}=\log\left(\sum_\alpha\vert\eta_\alpha\vert\right)$
and each $\cO_\alpha$ is a CPTP. Substituting this decomposition into~\eqref{eq:superchannel}
yields
\begin{align}
\Theta(\cN)
&= \cP^{\text{post}}_{AE\to A}\circ\left(\cN_{A\to A}\ox \id_{E}\right)
            \circ \cP^{\text{pre}}_{A\to AE} \\
&= \sum_{\alpha}\eta_{\alpha}
    \cP^{\text{post}}_{AE\to A}\circ\left(\cO_\alpha\ox \id_{E}\right)
            \circ \cP^{\text{pre}}_{A\to AE} \\
&= \sum_{\alpha}\eta_{\alpha}\cO^\prime_\alpha,
\end{align}
where each $\cO^\prime_\alpha\equiv\cP^{\text{post}}\circ\cO_\alpha
        \circ\cP^{\text{pre}}$ is a quantum channel.
This induces an valid linear decomposition of $\Theta(\cN)$ and thus
$\ce{\Theta(\cN)} \leq \ce{\cN}$.
\end{proof}

\begin{figure}[!htbp]
  \centering
  \includegraphics[width=0.6\textwidth]{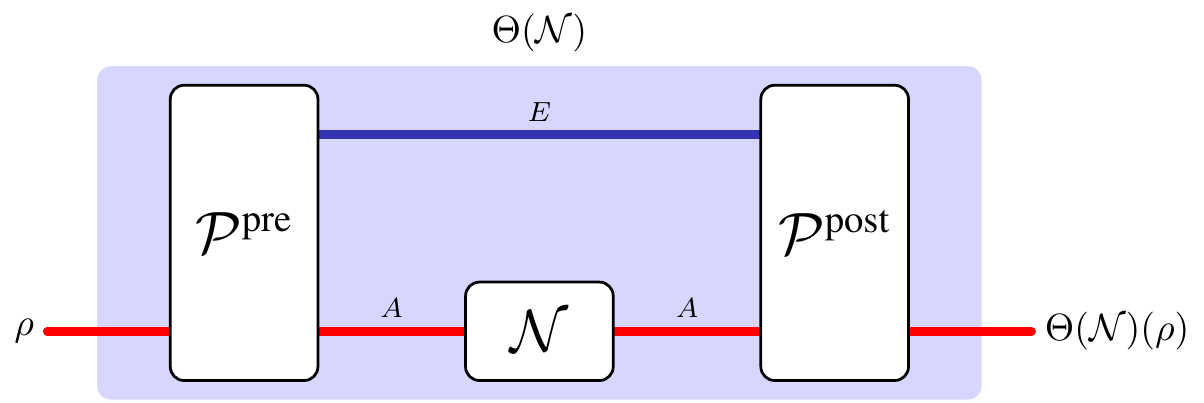}
  \caption{A superchannel transforms a quantum linear map to quantum linear map.}
  \label{fig:superchannel}
\end{figure}

This resource perspective motivates us to 
investigate the physical implementability of linear maps
within the quantum resource theory 
framework~\cite{chitambar2019quantum} and
explore the widely studied robustness measure as well as other channel resource measures~\cite{WW18,Diaz2018,Fang2018,WWS19,yuan2019universal,Fang2019a,Wang2019}.
We investigate in Appendix~\ref{sec:robustness} 
the proposed robustness measure 
and show that the two measures,
physical implementability and robustness,
are actually equivalent in some sense.

%%%%%%%%%%%%%%%%%%%%%%%%%%%%%%%%%%%%%%%%%%%%%%%%%%%%%%%%%%%%%%%%%%%%%%%%%%%%%%%%%%%%%%%%%%%%%%%%
\subsection{Bounds}

In this section, we derive lower and upper bounds on $\ce{\cN}$
in terms of the commonly used trace norm of the corresponding Choi operator $J_\cN$.
What's more, we show explicitly that these bounds are tight. Particularly, the bounds can be saturated by the inverse maps of some well-known quantum channels.

\begin{theorem}\label{thm:boundtracenorm}
Let $\cH$ be a $d$-dimensional Hilbert space and Let $\cN$ be an HPTP map in $\cH$.
It holds that
\begin{align}\label{eq:boundtracenorm}
    \frac{\norm{J_{\cN}}{1}}{d} \leq 2^{\ce{\cN}} \leq \norm{J_{\cN}}{1}.
\end{align}
\end{theorem}
\begin{proof}
To show the first inequality in~\eqref{eq:boundtracenorm}, we make use of the
primal SDP~\eqref{eq:r-SDP-Prime}. Notice that constraint~\eqref{eq:r-SDP-Prime-1} yields
\begin{align}\label{eq:nre}
  \norm{J_{\cN}}{1}
= \norm{J_1 - J_2}{1}
\leq \norm{J_1}{1} + \norm{J_2}{1}
= d(p_1+p_2),
\end{align}
where the inequality follows from the triangle inequality and the last equality
follows from the constraints~\eqref{eq:r-SDP-Prime-2} and~\eqref{eq:r-SDP-Prime-3}, respectively.
Since Eq.~\eqref{eq:nre} holds for arbitrary decomposition satisfying the constraints, it holds
in particular for the optimal decomposition and thus
\begin{align}
    \ce{\cN} = p_1 + p_2 \geq \norm{J_{\cN}}{1}/d.
\end{align}

To show the second inequality in~\eqref{eq:boundtracenorm}, recall that
 $J_\cN$ is Hermitian and thus diagonalizable. Consider the eigenvalue decomposition
\begin{align}
  J_\cN
= \sum_{i}\lambda_i\ketbra{\psi_i}{\psi_i}
= \sum_{i:\lambda_i\geq 0}\lambda_i\ketbra{\psi_i}{\psi_i}
 - \sum_{j:\lambda_j < 0}|\lambda_j|\ketbra{\psi_j}{\psi_j},
\end{align}
where in the second equality we group the eigenstates by the sign of the corresponding eigenvalues.
Let
\begin{align}
\eta_1 &:= \sum_{i:\lambda_i\geq 0}\lambda_i,\\
J_1 &:= \sum_{i:\lambda_i\geq 0}\frac{\lambda_i}{\eta_1}\ketbra{\psi_i}{\psi_i},\\
\eta_2 &:= \sum_{j:\lambda_j < 0}|\lambda_j|,\\
J_2 &:= \sum_{j:\lambda_j < 0}\frac{|\lambda_j|}{\eta_2}\ketbra{\psi_j}{\psi_j}.
\end{align}
By construction, we have $J_\cN=\eta_1J_1 - \eta_2J_2$,
$\eta_1,\eta_2\geq0$, $J_1,J_2\geq0$, and $\tr J_1 = \tr J_2 = 1$. What's more,
since $J_1\geq0$ and $\tr J_1 = 1$, the eigenvalues of $\tr_BJ_1$ are always less than $1$,
and thus $\tr_BJ_1\leq\1_A$. Similarly, $\tr_BJ_2\leq\1_A$.
By Lemma \ref{lem:N2-sub} we know the optimal value for decomposing $\cN$ onto CPTN and CPTP are the same, thus we can conclude
\begin{align}
    2^{\ce{\cN}} \leq \eta_1+\eta_2=\norm{J_\cN}{1}.
\end{align}
\end{proof}

\begin{remark}
As we will show later in Theorem~\ref{lem:unital},
the trace norm lower bound in~\eqref{eq:boundtracenorm} is tight and can be saturated
by the set of so-called mixed unitary linear maps. On the other hand, the trace upper bound can be saturated by the amplitude damping channel in the limit sense by Theorem~\ref{thm:AD channel}.
\end{remark}

%%%%%%%%%%%%%%%%%%%%%%%%%%%%%%%%%%%%%%%%%%%%%%%%%%%%%%%%%%%%%%%%%%%%%%%%%%%%%%%%%%%%%%%%%%%%%%%%
\subsection{Analytic expression for particular linear maps}\label{sec:Analytic expression}

In this subsection, we analytically evaluate  the physical implementability 
for some linear maps, which are the inverse map of some practically interesting quantum CPTP maps.

\paragraph*{Inverse map of the amplitude damping channel}
The qubit amplitude damping channel $\cA_{\epsilon}$ is given by
Kraus operators $A_0:=\proj{0}+\sqrt{1-\epsilon}\proj{1}$ and $A_1:=\sqrt{\epsilon}\ketbra{0}{1}$,
where $\epsilon\in[0,1]$. That is,
\begin{align}\label{eq:AD-channel}
    \cA_{\epsilon}(\rho) = A_0 \rho A_0^\dagger + A_1 \rho A_1^\dagger.
\end{align}
It turns out that $\cA_{\epsilon}$ is invertible, and the Choi operator of its
invertible map $\cA_{\epsilon}^{-1}$ has the form
\begin{align}
   J_{\cA_{\epsilon}^{-1}}=
   \begin{bmatrix}
1&0&0&\frac{1}{\sqrt{1-\epsilon}}\\
0&0&0&0\\
0&0&\frac{-\epsilon}{1-\epsilon}&0\\
\frac{1}{\sqrt{1-\epsilon}}&0&0&\frac{1}{1-\epsilon}
\end{bmatrix}.
\end{align}

\begin{theorem}\label{thm:AD channel}
For $\epsilon\in[0,1)$,
it holds that $\ce{\cA_{\epsilon}^{-1}} = \log\frac{1+\epsilon}{1-\epsilon}$.
What's more,
\begin{align}
    \lim_{\epsilon\to1}\ce{\cA_{\epsilon}^{-1}} 
    = \log\norm{J_{\cA_{\epsilon}^{-1}}}{1}.
\end{align}
\end{theorem}
\begin{proof}
We prove this theorem by exploring the primal and dual SDP characterizations
in Eqs.~\eqref{eq:r-SDP-Prime} and~\eqref{eq:r-SDP-Dual}.
Let $p_1 = 1/(1-\epsilon)$, $p_2 = \epsilon/(1-\epsilon)$,
\begin{align}
  J_1 :=
\begin{bmatrix}
\frac{1}{1-\epsilon} & 0 & 0 & \frac{1}{\sqrt{1-\epsilon}} \\
0&0&0&0\\
0&0&0&0\\
\frac{1}{\sqrt{1-\epsilon}}&0&0&\frac{1}{1-\epsilon}
\end{bmatrix},\quad
  J_2 :=    \begin{bmatrix}
\frac{\epsilon}{1-\epsilon} & 0 & 0 & 0 \\
0&0&0&0\\
0&0&\frac{\epsilon}{1-\epsilon}&0\\
0&0&0&0
\end{bmatrix}.
\end{align}
One can check that the tetrad $(p_1,p_2,J_1,J_2)$ is a feasible solution to
the primal SDP~\eqref{eq:r-SDP-Prime}, yielding
\begin{align}\label{eq:AD channel 1}
  2^{\ce{\cA_{\epsilon}^{-1}}} \leq p_1 + p_2 = \frac{1+\epsilon}{1-\epsilon}.
\end{align}

On the other hand, set
\begin{align}
    M_{AB} :=
\begin{bmatrix}
1 & 0 & 1 & 0 \\
0& 1 & 0 & 1 \\
1 & 0 & -2 & 0 \\
0 & 1 & 0 & 0
\end{bmatrix},\quad
   K_A :=
\begin{bmatrix}
1 & 1 \\
1 & 0
\end{bmatrix},\quad
   N_A :=
\begin{bmatrix}
-1 & -1 \\
-1 & 2
\end{bmatrix}.
\end{align}
We can verify that the triple $(A,B,D)$ is a feasible solution to
the dual SDP~\eqref{eq:r-SDP-Dual}, leading to
\begin{align}\label{eq:AD channel 2}
  2^{\ce{\cA_{\epsilon}^{-1}}} \geq \tr\left[M_{AB} J_{\cA_{\epsilon}^{-1}}\right]
= \frac{1+\epsilon}{1-\epsilon}.
\end{align}
Eqs.~\eqref{eq:AD channel 1} and~\eqref{eq:AD channel 2} together give the desired result.
\end{proof}

\begin{remark}
Note that the previous work~\cite{takagi2020optimal}
has investigated the non-physical implementability of $\cA_{\epsilon}^{-1}$
w.r.t. the set of so-called implementable operations\footnote{We refer to Eqs. (2) and (3)
in~\cite{takagi2020optimal} for the definition of implementable operations.}
by imposing both lower and upper bounds
on $2^{\ce{\cA_{\epsilon}^{-1}}}$~\cite[Theorem 3]{takagi2020optimal}, which has not been tight yet.
Our Theorem~\ref{thm:AD channel} further strengthens the results on mitigating the amplitude damping noise by concluding that
the obtained upper bound is actually optimal even if a larger free
set of quantum operations is allowed.
\end{remark}

\paragraph*{Inverse map of the generalized amplitude damping channel}
The generalized amplitude damping (GAD) channel is one of the realistic
sources of noise in superconducting quantum
processor~\cite{chirolli2008decoherence,Khatri2019}, whose quantum capacity has been studied in \cite{Khatri2019,Wang2019b}
It can be viewed as the qubit analogue of the bosonic thermal channel and can be used to
model lossy processes with background noise for low-temperature systems.
The generalized amplitude damping channel is a two-parameter family of channels
described as follows:
\begin{align}
\cA_{y, N}(\rho)
:= A_{1} \rho A_{1}^{\dagger}
 + A_{2} \rho A_{2}^{\dagger}
 + A_{3} \rho A_{3}^{\dagger}
 + A_{4} \rho A_{4}^{\dagger},
\end{align}
where $y,N\in[0,1]$ and
\begin{align}
A_{1}&:=\sqrt{1-N}(\ketbra{0}{0} + \sqrt{1-y}\ketbra{1}{1}),\\
A_{2}&:=\sqrt{y(1-N)}\ketbra{0}{1},\\
A_{3}&:=\sqrt{N}(\sqrt{1-y}\ketbra{0}{0}+\ketbra{1}{1}),\\
A_{4}&:=\sqrt{y N}\ketbra{1}{0}.
\end{align}
Note that when $N = 0$, $\cA_{y, 0}$ reduces to the conventional amplitude damping channel.
Similar to the amplitude damping channel, $\cA_{y, N}$ is invertible when $y\neq 1$,
and the Choi operator of its inverse map $\cA_{y, N}^{-1}$ has the form
\begin{align}
    J_{\cA_{y, N}^{-1}}=
    \begin{bmatrix}
        &\frac{1-y+Ny}{1-y} &0 &0 &\frac{1}{\sqrt{1-y}}\\
        &0 &\frac{-Ny}{1-y} &0 &0\\
        &0 &0 &\frac{-y+Ny}{1-y} &0\\
        &\frac{1}{\sqrt{1-y}} &0 &0 &\frac{1-Ny}{1-y}
    \end{bmatrix}.
 \end{align}

\begin{lemma}\label{lemma:generalized AD}
For $y\in[0,1)$, it holds that
\begin{align}\label{eq:generalized AD}
  \gamma(\cA_{y, N}^{-1})=\frac{1+|y-2Ny|}{1-y}.
\end{align}
\end{lemma}
\begin{proof}
We prove~\eqref{eq:generalized AD} by exploring 
the primal and dual SDP
characterizations in Eqs.~\eqref{eq:r-SDP-Prime} and~\eqref{eq:r-SDP-Dual}.
We divide the proof into two cases based on the value of $N$,
since it influences our constructions of the feasible solutions.

\textbf{Case 1: $N\leq 1/2$.} In this case, $y-2Ny \geq 0$.
We can design feasible solutions to both the primal~\eqref{eq:r-SDP-Prime}
and the dual~\eqref{eq:r-SDP-Dual} programs that both evaluate to $\frac{1+y-2Ny}{1-y}$:
\begin{align}
    J_1=\begin{bmatrix}
    \frac{1-Ny}{1-y} &0 &0& \frac{1}{\sqrt{1-y}}\\
    0 &0 &0 &0\\
    0 &0 &0 &0\\
    \frac{1}{\sqrt{1-y}} &0 &0 &\frac{1-Ny}{1-y}
    \end{bmatrix},
    \quad
    J_2=\begin{bmatrix}
    \frac{y-2Ny}{1-y} & 0 & 0 & 0\\
    0 &\frac{Ny}{1-y} & 0 & 0 \\
    0 & 0 &\frac{y-Ny}{1-y}& 0\\
    0& 0 & 0 &0
    \end{bmatrix},
\end{align}
\begin{align}
M_{AB}=
\begin{bmatrix}
1&0&1&0\\
0&1&0&1\\
1&0&-2&0\\
0&1&0&0
\end{bmatrix},
\quad
K_A=\begin{bmatrix}
1&1\\
1&0
\end{bmatrix},
\quad
N_A=\begin{bmatrix}
-1&-1\\
-1&2
\end{bmatrix}.
\end{align}

\textbf{Case 2: $N> 1/2$.} In this case, $2Ny - y \geq 0$.
We can design feasible solutions to both the primal~\eqref{eq:r-SDP-Prime}
and the dual~\eqref{eq:r-SDP-Dual} programs that both evaluate to $\frac{1-y+2Ny}{1-y}$:
\begin{align}
    J_1=\begin{bmatrix}
    \frac{1-y+Ny}{1-y} &0 &0& \frac{1}{\sqrt{1-y}}\\
    0 &0 &0 &0\\
    0 &0 &0 &0\\
    \frac{1}{\sqrt{1-y}} &0 &0 &\frac{1-y+Ny}{1-y}
    \end{bmatrix},
    \quad
    J_2=\begin{bmatrix}
    0 & 0 & 0 & 0 \\
    0 &\frac{Ny}{1-y} & 0 & 0 \\
    0 & 0 &\frac{y-Ny}{1-y}& 0 \\
    0 & 0 & 0 & \frac{2Ny-y}{1-y}
    \end{bmatrix},
\end{align}
\begin{align}
M_{AB}=
\begin{bmatrix}
1&0&1&0\\
0&-1&0&1\\
1&0&0&0\\
0&1&0&0
\end{bmatrix},
\quad
K_A=\begin{bmatrix}
1&1\\
1&0
\end{bmatrix},
\quad
N_A=\begin{bmatrix}
1&1\\
1&0
\end{bmatrix}.
\end{align}
\end{proof}

\paragraph*{The mixed unitary map.}
Let $\mathscr{U}=\{U_i\}_{i\in\cI}$ be a set of unitaries in $\cH$.
We say an HPTP map $\cT\in\hptp{\cH}$ is a \textit{mixed unitary map w.r.t. $\mathscr{U}$}
if there exists a set of real numbers $\{r_i\in\bR\}_{i\in\cI}$ such that
$\sum_{i\in\cI}r_i =1$ and
\begin{align}\label{eq:mixed unitary map}
   \cT(\cdot) := \sum_{i\in\cI}r_iU_i(\cdot)U_i^\dagger.
\end{align}
Note that the such defined mixed unitary map can be viewed a natural extension of
the mixed unitary channel intensively studied in~\cite[Chapter 4]{watrous2018theory},
by allowing negative coefficients in the mixture.
Interestingly, if $\mathscr{U}$ possesses the orthogonality property,
the non-physical implementability
of arbitrary mixed unitary map w.r.t. $\mathscr{U}$ can be evaluated analytically.

\begin{theorem}\label{lem:unital}
Let $\mathscr{U}=\{U_i\}_{i\in\cI}$ be a set of unitaries satisfying
the mutual orthogonality condition: $\forall i\neq j$,
$\tr[U_i^\dagger U_j] = 0$.
For arbitrary mixed unitary map $\cT$ w.r.t. $\mathscr{U}$ of the form~\eqref{eq:mixed unitary map},
it holds that
\begin{align}
  \ce{\cT} = \log\frac{\norm{J_\cT}{1}}{d} = \log\left(\sum_{i\in\cI}\vert r_i\vert \right).
\end{align}
where $d$ is the dimension of the system which $\cT$ is acting on.
\end{theorem}

\begin{proof}
By the definition of physical implementability, we have $\ce{\cT}\leq\log(\sum_i|r_i|)$.
In the following we show that $\norm{J_{\cT}}{1}=d\sum_i |r_i|$.
This, together with the lower bound in Theorem~\ref{thm:boundtracenorm}, concludes the proof.
Notice that
\begin{align}\label{eq:unital-1}
    J_{\cT}
:= (\id\otimes \cN)(\proj{\Gamma})
 = d\sum_i r_i (\1\otimes U_i)\proj{\Psi}(\1\otimes U_i^\dagger).
\end{align}
We claim $\{(\1\otimes U_i)\ket{\Psi}\}_{i\in\cI}$ is a set of orthogonal unit vectors, i.e.,
\begin{align}\label{eq:unital-2}
    \bra{\Psi}(\1\otimes U_i^\dagger)(\1\ox U_j)\ket{\Psi} =\delta_{ij}.
\end{align}
When $i=j$, the equation can be verified directly. When $i\neq j$,
it holds that
\begin{align}
    \bra{\Psi}(\1\otimes U_i^\dagger)(\1\ox U_j)\ket{\Psi}
&=  \frac{1}{d}\sum_{m,n=0}^{d-1}\bra{mm}(\1\otimes U_i^\dagger U_j) \ket{nn} \\
&=  \frac{1}{d}\sum_{m,n=0}^{d-1}\bra{m} U_i^\dagger U_j \ket{m}\\
&=  \frac{1}{d}\tr(U_i^\dagger U_j)\\
&=  0,
\end{align}
where the last equality follows from the orthogonality condition.
Eqs.~\eqref{eq:unital-1} and~\eqref{eq:unital-2} together
yield $\norm{J_{\cT}}{1}=d\sum_i |r_i|$.
We are done.
\end{proof}

\paragraph*{Invertible map of the qudit depolarizing channel.}
As an interesting implication of Theorem~\ref{lem:unital},
we can analytically derive  the non-physical implementability of
the invertible maps of both the depolarizing and dephasing channels.
In the following, we first formally define  the depolarizing and dephasing channel
in the $d$-dimensional Hilbert space $\cH$.
Let $\bZ_d:=\{0,1,\cdots,d-1\}$ which forms a
ring w.r.t. to addition and multiplication modulo $d$.
The set of discrete Weyl operators $\{W_{x,z}\}_{x,z\in\bZ_d}$ in $\cH$ is
defined as~\cite[Section 4.1.2]{watrous2018theory}
\begin{align}
  W_{x,z} := X^x Z^z,
\end{align}
where the generalized Pauli operators $X$ and $Z$ are defined as
\begin{align}\label{eq:generalized Pauli Z}
    X := \sum_{k\in\bZ_d}\ketbra{k+1}{k},\quad
    Z := \sum_{k\in\bZ_d}\zeta^k\proj{k},
\end{align}
with the $d$-root of unity $\zeta:=e^{2\pi i/d}$.
Notice that the set of Weyl operators satisfies the orthogonality condition:
\begin{align}\label{eq:orthogonality condition}
    \tr\left[W_{x,z}^\dagger W_{x',z'}\right]
= \begin{cases}
  d, &\text{ if } (x,z) = (x',z') \\
  0, &\text{ otherwise.}
\end{cases}
\end{align}
The qudit depolarizing quantum channel $\cD_{d,\epsilon}$, where $\epsilon\in[0,1]$, is defined
in terms of the Weyl operators as
\begin{align}\label{eq:depolarizing channel}
    \cD_{d,\epsilon}(\rho) := (1-\epsilon)\rho
        + \frac{\epsilon}{d^2}\sum_{x,z\in\bZ_d} W_{x,z}\rho W_{x,z}^\dagger.
\end{align}

In \cite{takagi2020optimal}, Takagi gave the optimal sampling cost for qudit quantum channel and qubit dephasing channel, using ad-hoc techniques. Here we show that those bounds can be derived from Theorem \ref{lem:unital}.

Regarding the depolarizing channel of great practical interests, we have the following.

\begin{lemma}\label{lemma:depolarizing}
For $\epsilon\in[0,1)$,
it holds that $\ce{\cD_{d,\epsilon}^{-1}}=\log\frac{1+(1-2/d^2)\epsilon}{1-\epsilon}$.
\end{lemma}
\begin{proof}
First of all, notice that~\cite[Theorem 1]{takagi2020optimal}
\begin{align}
  \cD_{d,\epsilon}^{-1}(\rho)
= \left(1+\frac{(d^2-1)\epsilon}{d^2(1-\epsilon)}\right)\id(\rho) -
    \frac{\epsilon}{d^2(1-\epsilon)}
    \sum_{(x,z)\in\bZ_d\times \bZ_d\backslash(0,0)} W_{x,z}\rho W_{x,z}^\dagger
\end{align}
and thus $\cD_{d,\epsilon}^{-1}$ is a mixed unitary map w.r.t.
the set of Weyl operators $\{W_{x,z}\}$.
Since this set satisfies the orthogonality condition~\eqref{eq:orthogonality condition},
Theorem~\ref{lem:unital} implies that
\begin{align}
  \ce{\cD_{d,\epsilon}^{-1}}
= \log\left(1+\frac{(d^2-1)\epsilon}{d^2(1-\epsilon)} +
    (d^2-1)\frac{\epsilon}{d^2(1-\epsilon)}\right)
= \log\frac{1+(1-2/d^2)\epsilon}{1-\epsilon}.
\end{align}
\end{proof}

\paragraph*{Invertible map of the qubit dephasing channel.}
Let $\sigma_z$ be the Pauli Z operator.
Notice that when $d=2$, the operator $Z$ define in~\eqref{eq:generalized Pauli Z} reduces
to $\sigma_z$.
The qubit dephasing quantum channel $\cF_\epsilon$, where $\epsilon\in[0,1]$, is defined as
\begin{align}
    \cF_\epsilon(\rho) := (1-\epsilon)\rho + \epsilon \sigma_z \rho \sigma_z.
\end{align}
We have the following.
\begin{lemma}\label{lemma:dephasing}
For $\epsilon\in[0,1/2)$, it holds that $\ce{\cF_\epsilon^{-1}}=\log\frac{1}{1-2\epsilon}$.
\end{lemma}
\begin{proof}
The proof follows similarly the argument of Lemma~\ref{lemma:depolarizing}
by noticing that~\cite[Theorem 1]{takagi2020optimal}
\begin{align}
  \cF_{\epsilon}^{-1}(\rho)
=\frac{1-\epsilon}{1-2\epsilon}\id(\rho)-\frac{\epsilon}{1-2\epsilon}\sigma_z(\rho)\sigma_z
\end{align}
and that the operators $\1$ and $\sigma_z$ are orthogonal.
\end{proof}
\begin{remark}
Unlike Lemma~\ref{lemma:depolarizing}, Lemma~\ref{lemma:dephasing} does not hold in general
for the qudit dephasing channel
$\cF_{d,\epsilon}(\rho) := (1-\epsilon)\rho + \epsilon Z \rho Z^\dagger$, 
where $Z$ is define in~\eqref{eq:generalized Pauli Z}.
This dues to that $Z$ does not satisfy the relation $Z^\dagger = Z$
in general. It remains as an interesting problem to compute analytically
$\ce{\cF_{d,\epsilon}^{-1}}$ for $d\geq3$.
\end{remark}

\begin{remark}
We emphasize that Lemma~\ref{lemma:depolarizing} and Lemma~\ref{lemma:dephasing}
have been shown previously in Theorems 1 and 2 of~\cite{takagi2020optimal}, respectively.
In that Ref., the author assumed that the linear map is decomposed w.r.t. the set of
implementable operations, which in turn is a strict subset of the set of quantum channels.
In this sense, our obtained results
% in Lemmas~\ref{lemma:depolarizing} and~\ref{lemma:dephasing}
enhance the previous ones by specifying the \emph{fundamental} limit on the
physical implementability of these two linear maps.
\end{remark}

\section{Applications in error mitigation}\label{sec:error mitigation}

In this section we endow the proposed physical implementability measure with
an operational interpretation within the quantum error mitigation framework
as it %quantifies the \emph{ultimate} sampling cost achievable using the full expressibility of quantum computers.
establishes the lower bound of the sampling cost achievable via the quasiprobability decomposition technique. 

\subsection{Physical implementability is the  sampling cost}

In quantum computing, especially in the NISQ era~\cite{preskill2018quantum},
a common computational task is to estimate the 
expected value $\tr[\rho A]$ for
a given observable $A$ and a quantum state $\rho$. Without loss of generality, we may assume $A$ to be diagonal in the computational basis, otherwise one can apply a unitary to $\rho$ firstly and then measure in the computational basis. That is, we consider measurement  in the form of
\begin{align}
    A=\sum_{x\in\{0,1\}^n}A(x)\ketbra{x}{x},\; A(x)\in [-1,1].
\end{align}

If $\rho$ can be prepared perfectly, one can get $\tr(\rho A)$ directly by a sequence of measurements. 
However, the preparation of $\rho$ inevitably suffers
from noise that can be modeled by some CPTP $\cO\in\channel{\cH}$,
rendering the value $\tr[\cO(\rho) A]$.
\textit{How can we deal with the noise and recover the expected value anyway?} There are many recently proposed error mitigation methods to accomplish this task (see, e.g., \cite{temme2017error,endo2018practical,takagi2020optimal,Bonet-Monroig2018,Endo2020,Wang2021,Cai2021}).

Since the preparation procedure is given \emph{a priori}, one feasible way is to perform the
invertible map $\cO^{-1}$ (assumed to exist), yielding
\begin{align}\label{eq:invertible}
  \tr\left[\cO^{-1}\circ\cO(\rho)A\right] = \tr[\rho A],
\end{align}
which successfully \JJQ{mitigates the noise}.
This is quite similar to the quantum channel correction task in the first glance, since we can think of $\cO^{-1}$ as a correcting procedure.

However, the main problem with~\eqref{eq:invertible} is that the invertible map $\cN\equiv\cO^{-1}$
might not be physically implementable, i.e., it is not a CPTP, though
we have already shown in Property~\ref{pro:hermitian} that $\cN$ is
both Hermitian- and trace-preserving. Thus the problem that we are faced with is to
physically approximate the effect of an HPTP map, which would unavoidably incur
extra computation cost due to approximation. 
We use the probabilistic error cancellation technique~\cite{crowder2015linearization,temme2017error,howard2017application,endo2018practical,takagi2020optimal} to deal with this issue
and it turns out that the incurred computation cost is quantified exactly by the physical
implementability measure. The precise error mitigation procedure goes as follows:
\begin{enumerate}
  \item We optimally decompose $\cN$ into a linear 
        combination of CPTP maps
        $\{\eta_\alpha,\cO_\alpha\}_\alpha$
        as~\eqref{eq:implementability}. Then
        $2^{\ce{\cN}} = \sum_\alpha\vert\eta_\alpha\vert$.
  \item We iterate the following sampling procedure $M$ times.
        In the $m$-th round where $m\in[M]$,
      \begin{enumerate}[2.1.]
        \item We sample a CPTP map $\cO^{(m)}$ from $\{\cO_\alpha\}_\alpha$
              with probability $\{\vert\eta_\alpha\vert/\sum_\alpha\vert\eta_\alpha\vert\}_\alpha$.
              Denote by $\eta^{(m)}$ the sampled coefficient.
        \item Apply CPTP $\cO^{(m)}$ to $\cO(\rho)$, measure each qubit in the computational basis. Denote $s^{(m)}\in\{0,1\}^n$ as the
        binary string obtained and $A\left(s^{(m)}\right)$ as the measurement value. Write a random variable \begin{align}
            X^{(m)}=2^{\ce{\cN}}\opn{sgn}\left(\eta^{(m)}\right)
                   A\left(s^{(m)}\right) \in[-2^{\ce{\cN}},2^{\ce{\cN}}].
        \end{align}
              where $\opn{sgn}:\bR\to\pm1$ is the sign function defined as
        $\forall x\leq0$, $\opn{sgn}(x)=-1$ and $\forall x>0$, $\opn{sgn}(x)=1$.
      \end{enumerate}
  \item Using the data obtained in the second step, we compute the       following \emph{empirical mean value}
        \begin{align}
          \xi &:= \frac{1}{M} \sum_{m=1}^M X^{(m)} =\frac{2^{\ce{\cN}}}{M}\sum_{m=1}^M\opn{sgn}\left(\eta^{(m)}\right)
                   A\left(s^{(m)}\right) \label{eq:empirical mean value}
        \end{align}
  \item Output $\xi$ as an estimation of the target expected value $\tr[\rho A]$.
\end{enumerate}

Now we analyze the efficiency and accuracy of the above estimation procedure.
% Theorem 2 of Hoeffding (1963) in wiki
First of all, we show in Lemma~\ref{lemma:unbiased} that $\xi$ is actually an unbiased
estimator of $\tr[\rho A]$. This justifies the validity of the estimator, as long as $M$
is sufficiently large, due to the weak law of large numbers.
What's more, we can apply the Hoeffding inequality to ensure
that $M=2^{2\nu+1}\log(2/\varepsilon)/\delta^2$ number of samples in Step 2
would estimate the target expectation value $\tr[\rho A]$ within error $\delta$ with probability no less than $1-\varepsilon$:
\begin{align}
&\quad \Pr\left\{\vert\xi-\tr[\rho A]\vert\geq\delta\right\}
\stackrel{(a)}{\leq} 2\exp\left(-\frac{2M^2\delta^2}{4M2^{2\nu}}\right) \leq \varepsilon \\
\Rightarrow&\quad M \geq 2^{2\nu+1}\log(2/\varepsilon)/\delta^2,
\label{eq:number of copies}
\end{align}
where $(a)$ follows from the fact that $\vert X^{(m)}\vert\leq2^\nu$.
We can thus identify the unique role of $2^\nu$ in quantifying the number of rounds required to reach desired estimating precision, empowering the mathematically defined implementability
measure $\ce{\cN}$ an interesting operational meaning.
\begin{lemma}\label{lemma:unbiased}
The random variable $\xi$ defined in~\eqref{eq:empirical mean value} is an unbiased estimator of $\tr[\rho A]$.
\end{lemma}
\begin{proof}
Denote by $E(X)$ the expectation of a random variable $X$.
By ~\eqref{eq:empirical mean value} we have
\begin{align}
    E(\xi)&= \frac{2^{\ce{\cN}}}{M}\sum_{m=1}^M E\left(\opn{sgn}\left(\eta^{(m)}\right)
                   A\left(s^{(m)}\right)\right)\\
                   &= \frac{2^{\ce{\cN}}}{M}\sum_{m=1}^M \sum_\alpha \frac{|\eta_\alpha|}{2^{\ce{\cN}}}
         \opn{sgn}\left(\eta_{\alpha}\right)
                   \sum_{s\in\{0,1\}^n}  
                    \tr\left[\cO_\alpha\circ \cO(\rho)\ketbra{s}{s} \right] A\left(s\right)\\
                   &=\frac{2^{\ce{\cN}}}{M} \sum_{m=1}^M \sum_\alpha \frac{\eta_\alpha}{2^{\ce{\cN}}}
                   \tr\left[\cO_\alpha \circ \cO(\rho)A \right]\\
                   &=\frac{2^{\ce{\cN}}}{M} \sum_{m=1}^M \frac{1}{2^{\ce{\cN}}}\tr\left[\cN \circ \cO(\rho)A \right]\\
                   &=\tr\left[\rho A \right].
\end{align}
\end{proof}

\vspace{0.1in}

One may wonder that the proposed error mitigation setting is a bit unrealistic as the assumption
of one can implement all CPTPs perfectly directly remove the necessity of error mitigation in the very first place.
To deal with this concern, we describe the setting
in more details. In our error mitigation setting,
the state preparation and the error mitigation procedures 
are performed by different parties.
Namely, user $A$ will prepare the quantum state $\rho$ 
subject to inevitable quantum noise $\mathcal{O}$, 
while user $B$ can perform error mitigation on the received quantum state $\mathcal{O}(\rho)$ noiselessly. 
Since user $B$ only knows the error 
model $\mathcal{O}$ but has no information about 
the ideal state $\rho$ that user $A$ aims to prepare,
he \JJQ{cannot} prepare $\rho$ directly by himself.  

At the first glance, our setting  is 
different from the more relevant setting where \emph{all}
quantum operations are subject to noise. 
However, these two settings are in fact closely related, 
argued as follows. 
For any invertible quantum noise $\mathcal{O}$, 
we decompose $\mathcal{O}^{-1}$ into a linear
combination of CPTPs as
\begin{align}
    \mathcal{O}^{-1}:=\sum_{\alpha} \eta_\alpha \mathcal{O}_\alpha
\end{align}
Consider the setting where every operation is subject to the noise $\mathcal{O}$. The target is to implement a noiseless
quantum operation $\mathcal{U}$. Notice that
\begin{align}
\mathcal{U} 
=\mathcal{O}\mathcal{O}^{-1}\mathcal{U}
=\sum_\alpha \eta_\alpha \mathcal{O}\mathcal{O}_\alpha\mathcal{U},
\end{align}
inspiring a feasible method to implement $\mathcal{U}$ 
ideally: We firstly sample $\alpha$, and then perform 
the quantum operation $\mathcal{O}_\alpha \mathcal{U}$. 
This quantum operation will be inevitably corrupted 
by the noise $\cO$.
Statistically, the operation we effectively 
perform is exactly $\mathcal{U}$.

%\subsection{Properties of fundamental error mitigation cost}
\subsection{Properties of the above error mitigation procedure}

We explore the operational properties of this error mitigation procedure based on the nice properties $\nu$ and its connection to the sampling cost: 

\begin{itemize}
\item First, since the physical implementability measure $\nu$ 
        can be computed 
        efficiently via semidefinite programs (cf. Theorem~\ref{thm:N2}), 
        we can estimate the sampling cost of arbitrary linear maps, 
        yielding a feasible way to deal with quantum noise in the NISQ era~\cite{preskill2018quantum}. 
        \JJQ{However, 
        since the overall sampling cost of a quantum circuit
        is given by the product of the sampling cost
        of each gate in the circuit,
        the overall cost in our error mitigation procedure has an exponential scaling in the number of quantum gates, indicating that error mitigation cannot substitute the role of error correction.  
        Indeed, the proposed error mitigation is an auxiliary method to alleviate quantum errors 
        since NISQ devices do not have enough qubits to support error correcting codes.}
\item Second, the additivity of $\nu$ w.r.t. 
        tensor product of linear maps (cf. Theorem~\ref{thm:additivity})
        implies that for parallel quantum noises, global error mitigation has \textit{no advantage} over local error mitigation,
        i.e., dealing with quantum noises individually.
        On the other hand, the 
        subadditivity of $\nu$ w.r.t. 
        composition of linear maps (cf. Theorem~\ref{thm:composition-subadditivity})
        implies that for sequential quantum noises,
        treating them as a whole might be beneficial and reduce
        the sampling cost, compared to handling these noises
        one by one.
        This is intuitive since these quantum noises might cancel
        mutually in the sequential procedure.
\item Third, we obtain  the sampling cost analytically
        for some linear maps which are the inverse map of 
        practically interesting quantum CPTP maps.
        Prominent examples include the amplitude damping channel 
        (cf. Theorem~\ref{thm:AD channel}),
        the generalized amplitude damping channel
        (cf. Lemma~\ref{lemma:generalized AD}),
        the qudit depolarizing channel
        (cf. Lemma~\ref{lemma:depolarizing}),
        and the qubit dephasing channel
        (cf. Lemma~\ref{lemma:dephasing}).
        These results shed lights on dealing with 
        quantum noises in the NISQ era since they 
        remain as the lower bound of the sampling cost which quasiprobability method may achieve.
        %ultimate bound on which the best mitigation technique that may achieve.
\end{itemize}

\begin{comment}
Practically, if a CPTP map $\cN$ can be error mitigated,
we think that it is not \emph{very} noisy, since it's output
can be corrected in some way. 
This leads us to wonder if it is able to carry out
quantum communication, i.e., it has positive quantum capacity.
On the contrary, if $\cN$
is not invertible, i.e., $\cN^{-1}$ does not exist in the sense
of~\eqref{equ:inverse}, we are not able to carry out
the proposed error mitigation procedure.
In this case, one may think $\cN$ cannot be utilized to 
do quantum communication, since its output might not be
corrected.
In Appendix~\ref{appx:quantum communication},
we explore this intuition in depth by
investigating the relationship between invertibility and 
the ability of quantum communication of a given qubit CPTP map.
Our results reveal that these two concepts,
invertibility and the ability of quantum communication, 
are possibly different.
\end{comment}

\section{Conclusions and discussions}\label{sec:conclusion}

\JJQ{In this work, we ask and study the question how can one simulate the action of a general linear map on a quantum state when only quantum operations (CPTP maps) are available. We apply the quasiprobability sampling method and use mathematical tools from semidefinite programming to answer this  question, which leads to the first operational quantification of the physical implementability (or non-physicality from another perspective) for general linear maps.}

We offered a systematic way to approximate a general linear map, that may not be physically implementable, by decomposing it into a linear combination of physically implementable quantum operations, mathematically characterized by completely positive and trace-preserving maps (CPTPs), motivated by the appealing quasi-probability decomposition technique. We introduced the physical implementability measure $\nu(\mathcal{N})$ of a linear map $\cN$, which is the least amount of negative portion in the quasi-probability decomposition, as a quantifier of 
how well $\cN$ can be approximated by CPTPs. 
We show that $\nu$ is efficiently computable by semidefinite programs.
We proved that $\nu$ satisfies many interesting properties such as 
faithfulness, additivity with respect to the tensor product, 
and unitary channel invariance. 
We also derived upper and lower bounds of $\nu$ 
based on the trace norm of the target linear map's Choi operator, 
and obtained analytic expressions for several practical linear maps. 
Finally, we empowered this measure an operational meaning in the
quantum error mitigation task by showing that it %quantifies the ultimate sampling cost achievable using the full expressibility of quantum computers.
establishes the lower bound of the sampling cost achievable  via the quasiprobability decomposition technique. 

We expect that our proposed framework can find more applications
in quantum information and quantum computation. 
It is also interesting to further explore 
the structure of invertible qubit HPTP maps and derive an analytic expression
for the physical implementability measure.
In Section~\ref{sec:error mitigation}, 
we contributed an efficient method to estimate 
the expected value $\tr[\rho A]$ in the presence of noise characterized
by some noisy quantum channel $\cN$. 
This is an important task in quantum computing,
especially in the NISQ era. 
We yearn for new and novel methods for  this task. 

\paragraph{Acknowledgements.}
J. J. and K. W. contributed equally to this work. This work was done when J. J. was a research intern at Baidu Research.

\addcontentsline{toc}{section}{References}
\bibliographystyle{unsrtnat}
\bibliography{ref.bib}

\begin{appendices}

\section{Robustness measure}\label{sec:robustness}

\subsection{Definition}
As argued around Theorem~\ref{thm:Monotonicity},
the set of CPTP maps is treated as the free set
when defining the physical implementability measure from the resource theoretic perspective.
This motivates us to consider this problem within the quantum resource theory
framework~\cite{chitambar2019quantum} and explore the intensively investigated
robustness measure~\cite{harrow2003robustness,vidal1999robustness,
steiner2003generalized,brandao2007entanglement,almeida2007noise,
takagi2019operational,piani2015necessary,napoli2016robustness,piani2016robustness,
anand2019quantifying,takagi2019general,yuan2019universal,liu2019resource,
bae2019more,takagi2019operational,skrzypczyk2019robustness,takagi2020application,takagi2020optimal}.
More precisely, Let $\cN$ be an HPTP map,
we define the (absolute) \emph{robustness of physical implementability} of $\cN$ as
\begin{align}\label{eq:RoC}
   R(\cN) := \min_{\cT\text{~is CPTP}}
              \left\{s\geq0\sbar \frac{\cN + s\cT}{1+s}\text{~is CPTP}\right\}.
\end{align}
Note the minimization is well defined since the completely depolarizing channel is free.
Intuitively, $R(\cN)$ quantifies how \emph{robust} the linear map $\cN$
is against any physical implementation.
Alternatively, we can express $R(\cN)$ in terms of its Choi operator as
\begin{subequations}\label{eq:RoC-SDP}
\begin{align}
   \text{\bf Primal:}\quad
   R(\cN) =  \min&\; s \\
              \text{s.t.}&\; J_{\cN} + sJ_{\cT} = (1+s) J_{\cK}\label{eq:RoC-SDP-1} \\
                        &\; J_{\cT} \geq 0, \tr_BJ_{\cT} = \1_A \label{eq:RoC-SDP-2}\\
                        &\; J_{\cK} \geq 0, \tr_BJ_{\cK} = \1_A.\label{eq:RoC-SDP-3}
\end{align}
\end{subequations}
We can further simplify the above program using the trace-preserving condition.
Assume the pair $(s,\cT,\cK)$ achivese $R(\cN)$ in~\eqref{eq:RoC-SDP}.
Set $\wt{J}:= (1+s)J_{\cK}$, then $\tr_B\wt{J}=(1+s)\1_A$ due to Eq.~\eqref{eq:RoC-SDP-3}.
Eq.~\eqref{eq:RoC-SDP-1} guarantees that
$\wt{J} - J_{\cN} = sJ_{\cT}\geq 0$, following from the fact that $s\geq0$ and Eq.~\eqref{eq:RoC-SDP-2}.
As so, we can simplify~\eqref{eq:RoC-SDP} as
\begin{subequations}\label{eq:RoC-SDP-Simplified}
\begin{align}
\text{\bf Simplified Primal:}\quad
   R(\cN) =     \min&\; s \\
         \text{s.t.}&\; \wt{J} \geq J_{\cN} \\
                    &\; \tr_B\wt{J} = (s+1)\1_A\label{eq:RoC-SDP-Simplified-2} \\
                    &\; \wt{J} \geq 0, s \geq 0
\end{align}
\end{subequations}
Correspondingly, the dual SDP is given by
\begin{subequations}\label{eq:RoC-SDP-Simplified-Dual}
\begin{align}
\text{\bf Simplified Dual:}\quad
   R(\cN) = \max&\; \tr\left[M_{AB}J_{\cN}\right] - 1 \\
            \text{s.t.}&\; \tr N_A = 1 \\
                       &\; M_{AB} \leq N_A\ox\1_B \\
                       &\; M_{AB} \geq 0
\end{align}
\end{subequations}
One may check that the above SDP satisfies strong duality by the
Slater's theorem~\cite{watrous2017semidefinite}.

\subsection{Relation with the physical implementability}

It turns out that the physical implementability measure $\nu$~\eqref{eq:implementability}
is closely related to the robustness of physical 
implementability $R$~\eqref{eq:RoC},
resembling the relations previously obtained 
in~\cite[Eq. (1)]{howard2017application}
and~\cite[Eq. (6)]{takagi2020optimal}.

\begin{theorem}
Let $\cN$ be an HPTP map. It holds that
\begin{align}
  2^{\ce{\cN}} = 2R(\cN) + 1.
\end{align}
\end{theorem}

\begin{proof}
Note that this theorem can be proved using a similar technique
presented in~\cite[Appendix A]{takagi2020optimal}.
Here we write down the proof procedure for completeness.

``$\leq$'': Assume the channel pair $(\cT,\cK)$ achieves $R(\cN)$, i.e.,
\begin{align}
    \frac{\cN + R(\cN)\cT}{1+ R(\cN)} = \cK.
\end{align}
Rearranging the elements leads to $\cN = (1+ R(\cN))\cK - R(\cN)\cT$,
yielding a feasible decomposition of $\cN$. As so, we obtain from Theorem~\ref{thm:N2}
that
\begin{align}
    2^{\ce{\cN}} \leq 1+ R(\cN) + R(\cN) = 1 + 2R(\cN).
\end{align}

``$\geq$'': Assume the ensemble $\{(\eta_\alpha,\cO_\alpha)\}_{\alpha\in\cX}$
achieves $\ce{\cN}$~\eqref{eq:implementability}.
Let $\cX^+$ be the collection of symbols for which the sign of $\eta_\alpha$ is positive
and similarly for $\cX^-$. We have $\cX^+\cup\cX^-=\cX$ and $\cX^+\cap\cX^-=\emptyset$.
Set $\eta^+:=\sum_{\alpha\in\cX^+}\vert\eta_\alpha\vert$
and $\eta^-:=\sum_{\alpha\in\cX^-}\vert\eta_\alpha\vert$.
By assumption, $2^{\ce{\cN}}=\eta^++\eta^-$. Since $\cN$ is trace-preserving, we also have
$\eta^+-\eta^-=1$. We can divide the channels
into two groups according to the sign of their coefficients:
\begin{align}
    \cN = \sum_{\alpha\in\cX} \eta_\alpha\cO_\alpha
&= \sum_{\alpha\in\cX^+}\vert\eta_\alpha\vert\cO_\alpha
 - \sum_{x\in\cX^-}\vert\eta_\alpha\vert\cO_\alpha \\
&= \eta^+\left(\sum_{x\in\cX^+}\frac{\vert\eta_\alpha\vert}{\eta^+}\cO_\alpha\right)
  - \eta^-\left(\sum_{x\in\cX^-}\frac{\vert\eta_\alpha\vert}{\eta^-}\cO_\alpha\right) \\
&= (1+\eta^-)\cT - \eta^-\cK,
\end{align}
where $\cT:=\sum_{x\in\cX^+}\eta_\alpha/\eta^+\cO_\alpha$
and $\cK:=\sum_{x\in\cX^-}\vert\eta_\alpha\vert/\eta^-\cO_\alpha$ are well-defined quantum channels.
This gives $\cR(\cN)\leq \eta^- = (2^{\ce{\cN}}-1)/2$. We are done.
\end{proof}

\section{Proof of Eq.~\eqref{eq:r-SDP-Dual}}\label{appx:proof-of-the-dual}

In this Appendix, we derive the dual program for the primal program given in~\eqref{eq:r-SDP-Prime}.
Recall the primal SDP
\begin{subequations}\label{eq:tmp-SDP}
\begin{align}
    -2^{\ce{\cN}} = \max&\quad -(p_1+p_2) \\
      \text{s.t.}&\quad J_{\cN} = J_1 - J_2 \\
                 &\quad \tr_B J_1 = p_1\1_A\label{eq:tmp-SDP-c1} \\
                 &\quad \tr_B J_2 = p_2\1_A\label{eq:tmp-SDP-c2} \\
                 &\quad J_1, J_2 \geq 0
\end{align}
\end{subequations}

Introducing the Lagrange multipliers $M_{AB}\in\herm{AB}$ and $N_A,K_A\in\herm{A}$,
the Lagrange function of this primal SDP is given by
\begin{align}
  &\;   L(M_{AB},N_A,K_A) \\
:=&\; -(p_1 + p_2) + \langle M, J_{\cM} - J_1 + J_2\rangle
               + \langle N, p_1\1_A - \tr_B J_1 \rangle
               + \langle K, p_2\1_A - \tr_B J_2 \rangle \\
=&\; \langle M_{AB}, J_{\cM}\rangle + p_1(\tr[N_A]-1) + p_2(\tr[K_A]-1) \\
 &\qquad   + \langle J_1, - M_{AB} - N_A\ox\1_B \rangle
    + \langle J_2, M_{AB} - K_A\ox\1_B \rangle.
\end{align}
Since $J_1\geq0$, it must hold that $- M_{AB} - N_A\ox\1_B\leq 0$ 
otherwise the inner norm is unbounded.
Similarly, we have $M_{AB} - K_A\ox\1_B\leq 0$, $\tr[N_A] \leq 1$, 
and $\tr[K_A] \leq 1$. This leads to the dual SDP
\begin{subequations}
\begin{align}
 - 2^{\ce{\cN}} = \min&\quad \tr[M_{AB} J_{\cM}] \\
      \text{s.t.}&\quad \tr[N_A] \leq 1  \\
                 &\quad \tr[K_A] \leq 1 \\
                 &\quad M_{AB} + N_A\ox\1_B \geq 0 \\
                 &\quad - M_{AB} + K_A\ox\1_B \geq 0
\end{align}
\end{subequations}
Changing $\min$ to $\max$, we get 
\begin{subequations}
\begin{align}
  2^{\ce{\cN}} = \max&\quad -\tr[M_{AB} J_{\cM}] \\
      \text{s.t.}&\quad \tr[N_A] \leq 1  \\
                 &\quad \tr[K_A] \leq 1 \\
                 &\quad M_{AB} + N_A\ox\1_B \geq 0 \\
                 &\quad - M_{AB} + K_A\ox\1_B \geq 0
\end{align}
\end{subequations}
Since $M_{AB}$ is Hermitian, so is $- M_{AB}$. 
Substituting $M_{AB}$ with $-M_{AB}$ and
renaming the variables $N_A, K_A$, we can rewrite the above program as
\begin{subequations}\label{eq:r-SDP-Dual-inequality}
\begin{align}
2^{\ce{\cN}} = 
             \max&\quad \tr[M_{AB} J_{\cM}] \\
      \text{s.t.}&\quad \tr[N_A] \leq 1 \label{eq:r-SDP-Dual-inequality1} \\
                 &\quad \tr[K_A] \leq 1 \label{eq:r-SDP-Dual-inequality2} \\
                 &\quad M_{AB} + N_A\ox\1_B \geq 0
                    \label{eq:r-SDP-Dual-inequality3} \\
                 &\quad - M_{AB} + K_A\ox\1_B \geq 0 
                    \label{eq:r-SDP-Dual-inequality4}
\end{align}
\end{subequations}

Comparing~\eqref{eq:r-SDP-Dual-inequality} with~\eqref{eq:r-SDP-Dual},
we are left to show that the 
inequalities in~\eqref{eq:r-SDP-Dual-inequality1}
and~\eqref{eq:r-SDP-Dual-inequality2} can be further restricted to equalities. This is true since for any feasible $(M_{AB},N_A,K_A)$, 
we can reset it to be $(M_{AB},N_A+(1-\tr[N_A])\1_A/d_A,N_A+(1-\tr[N_A])\1_A/d_A)$. This new solution will make \eqref{eq:r-SDP-Dual-inequality1}
and~\eqref{eq:r-SDP-Dual-inequality2} to be equality. At the same time, it satisfies constraints \eqref{eq:r-SDP-Dual-inequality3}-\eqref{eq:r-SDP-Dual-inequality4} and  keep the objective value unchanged.

% \section{Relationship between two settings for error mitigation}\label{sec:relation}

% Recall that in our setting for error mitigation in Section \ref{sec:error mitigation}, 
% the  preparation of state $\rho$  is with noise $\mathcal{O}$ and the error mitigation for $\mathcal{O}(\rho)$ is noiseless.
% hile in many other  error mitigation literature, operations in error mitigation is also subject to noise. In the following we would show the two settings are related. 

% Observe that for any invertible noise $\mathcal{O}$, if in our setting,  the  $\mathcal{O}^{-1}$ is decomposed into CPTPs
% \begin{align}
%     \mathcal{O}^{-1}:=\sum_{\alpha} \eta_\alpha \mathcal{O}_\alpha
% \end{align}
% Then in the setting where every operation is subject to the noise $\mathcal{O}$,  to implement any noiseless CPTP operation $\mathcal{U}$, notice that
% \begin{align}
%     \mathcal{U}&=\mathcal{O}\mathcal{O}^{-1}\mathcal{U}\\
%     &=\sum_\alpha \eta_\alpha \mathcal{O}\mathcal{O}_\alpha\mathcal{U}
% \end{align}
% That is, to implement $\mathcal{N}$ noiseless, use ideas in probabilistic error cancellation, we can firstly sample $\alpha$, and then preform CPTP $\mathcal{O}_\alpha \mathcal{U}$. Since every operation is corrupted by noise, the operation we really perform is exactly the noiseless $\mathcal{U}$.

\end{appendices}

\end{document}